\documentclass[journal]{IEEEtran}

\usepackage{bbding}
\usepackage{booktabs}
\usepackage[table, dvipsnames]{xcolor}
\usepackage{soul}
\usepackage{changepage}
\usepackage{algorithmic}
\usepackage{algorithm2e}
\usepackage{float}

\usepackage{bbding}

\usepackage{glossaries}
\usepackage{pdflscape}
\usepackage{tikz}
\usepackage{amssymb}
\usepackage{cryptocode}
\usetikzlibrary{shapes}

\usepackage{multirow}
\usepackage{graphicx}
\usepackage{tabularx}
\usepackage{rotating}

\usepackage{caption}
\usepackage{subcaption}

\usepackage{amssymb,amsmath,amsthm}
\usepackage{float}
\usepackage{pgfplots}
\newtheorem{theorem}{Theorem}
\newtheorem{definition}{Definition}

\newcommand{\system}{HERMES\@}
\newcommand{\exprotocol}{SePCAR\@}

\newcommand\y{71}
\newcommand\sy{64}

\newcommand{\CC}{C\nolinebreak\hspace{-.05em}\raisebox{.4ex}{\tiny\bf +}\nolinebreak\hspace{-.12em}\raisebox{.4ex}{\tiny\bf +}}
\def\CC{{C\nolinebreak[4]\hspace{-.05em}\raisebox{.4ex}{\tiny\bf ++}}}

\newcounter{pr}
\newcommand{\pr}{\addtocounter{pr}{1}\arabic{pr} }

\newcounter{sr}
\newcommand{\sr}{\addtocounter{sr}{1}\arabic{sr} }
 
\newcounter{fr}
\newcommand{\fr}{\addtocounter{fr}{1}\arabic{fr} }

\newcounter{psr}
\newcommand{\psr}{\addtocounter{psr}{1}\arabic{psr} }

\def \totalMiMCsingleVehTime {127,37}
\def \totalMiMCmultiVehTime {137,44}

\def \stepAtime {52,7}

\def \stepBMiMCsingleVehTime {1,83}
\def \stepBMiMCmultiVehTime {11,9}
\def \stepBMiMCmultiVehATs {84}
\def \stepBMiMCmoreAT {696}
\def \stepBmacCBCtime {30,3}
\def \stepBCBCmoreAT {42}

\def \stepCtime {6,65}

\def \stepDtime {62,087}

\def \kdftime {2,87}

\def \shareshamirtime {10,78}

\def \Pubdectime {9,53}

\def \SymdecNexcomtime {3,15}

\def \signtime {4,25}

\def \signNexcomtime {32,43}
\def \signvrfNexcomtime {15,16}

\def \MACvrftime {3,49}

\newacronym{VSS}{VSS}{Vehicle Sharing-access System}
\newacronym{VSSP}{VSSP}{Vehicle Sharing-access Service Provider}
\newacronym{GDPR}{GDPR}{General Data Protection Regulation}
\newacronym{VM}{VM}{Vehicle Manufacturer}
\newacronym{OBU}{OBU}{On-Board Unit}
\newacronym{PL}{PL}{Public Ledger}
\newacronym{PD}{PD}{Portable Device}
\newacronym{MPC}{MPC}{Multiparty Computation}
\newacronym{PRF}{PRF}{Pseudo-Random Function}
\newacronym{PKI}{PKI}{Public Key Infrastructure}
\newacronym{EtA}{EtA}{Encrypt-then-Authenticate}
\newacronym{MAC}{MAC}{Message Authentication Code}
\newacronym{HSM}{HSM}{Hardware Security Module}
\newacronym{BD}{BD}{Booking Details}
\newacronym{AT}{AT}{Access Token}
\newacronym{RTT}{RTT}{Round Trip Time}
\newacronym{TP}{TP}{Trusted Provider}
\newacronym{SP}{SP}{Service Provider}
\newacronym{DB}{DB}{Database}

\newcommand{\length}[1]{||{#1}||}

\begin{document}
\title{\system: Scalable, Secure, and Privacy-Enhancing Vehicular Sharing-Access System} 

\author{
    \IEEEauthorblockN{
        Iraklis Symeonidis\IEEEauthorrefmark{1}, 
        Dragos Rotaru\IEEEauthorrefmark{2}\IEEEauthorrefmark{3},
        Mustafa A. Mustafa\IEEEauthorrefmark{4}\IEEEauthorrefmark{3}, 
        Bart Mennink\IEEEauthorrefmark{5}, 
        Bart Preneel\IEEEauthorrefmark{3},
        Panos Papadimitratos\IEEEauthorrefmark{1}}\\
\IEEEauthorblockA{
    \IEEEauthorrefmark{1}\textit{NSS}, \textit{KTH}, Sweden,
    \IEEEauthorrefmark{2}\textit{Cape Privacy},
    \IEEEauthorrefmark{3}\textit{imec-COSIC}, \textit{KU Leuven}, Belgium,
    \IEEEauthorrefmark{4}\textit{Department of Computer Science}, \textit{The University of Manchester}, UK,
    \IEEEauthorrefmark{5}\textit{Digital Security Group}, \textit{Radboud University}, Netherlands \\
    Email:  \IEEEauthorrefmark{1}irakliss@kth.se,  \IEEEauthorrefmark{2}dragos.rotaru@esat.kuleuven.be,  \IEEEauthorrefmark{3}mustafa.mustafa@manchester.ac.uk,  \IEEEauthorrefmark{5}b.mennink@cs.ru.nl,  \IEEEauthorrefmark{2}bart.preneel@esat.kuleuven.be,
    \IEEEauthorrefmark{1}papadim@kth.se
    }
}

\maketitle

\begin{abstract}
We propose \system, a scalable, secure, and privacy-enhancing system for users to share and access vehicles. \system\ securely outsources operations of vehicle access token generation to a set of untrusted servers. It builds on an earlier proposal, namely \exprotocol~\cite{DBLP:conf/esorics/SymeonidisAMMDP17}, and extends the system design for improved efficiency and scalability. To cater to system and user needs for secure and private computations, \system\ utilizes and combines several cryptographic primitives with secure multiparty computation efficiently. It conceals secret keys of vehicles and transaction details from the servers, including vehicle booking details, access token information, and user and vehicle identities. It also provides user accountability in case of disputes. Besides, we provide semantic security analysis and prove that \system\ meets its security and privacy requirements. Last but not least, we demonstrate that \system\ is efficient and, in contrast to \exprotocol, scales to a large number of users and vehicles, making it practical for real-world deployments. We build our evaluations with two different multiparty computation protocols: HtMAC-MiMC and CBC-MAC-AES. Our results demonstrate that \system\ with HtMAC-MiMC requires only \boldmath{$\approx \stepBMiMCsingleVehTime$}~ms for generating an access token for a single-vehicle owner and \boldmath{$\approx \stepBMiMCmultiVehTime$}~ms for a large branch of rental companies with over a thousand vehicles. It handles $546$ and $84$ access token generations per second, respectively. This results in \system\ being $\stepBMiMCmoreAT$ (with HtMAC-MiMC) and $\stepBCBCmoreAT$ (with CBC-MAC-AES) times faster compared to in \exprotocol\, for a single-vehicle owner access token generation. Furthermore, we show that \system\ is practical on the vehicle side, too, as access token operations performed on a prototype vehicle on-board unit take only $\approx \stepDtime$~ms.
\end{abstract}


\IEEEpeerreviewmaketitle

\section{Introduction}
Vehicle-sharing is an emerging smart mobility service leveraging connectivity and modern technology to enable users to share their vehicles with others. Users can book in advance and access a vehicle, where the traditional key distribution is naturally replaced. With the use of in-vehicle telematics and omnipresent \glspl{PD}, such as smartphones, vehicle owners can distribute (temporary) digital vehicle-keys, a.k.a. \glspl{AT}, to other users enabling them to access a vehicle~\cite{DBLP:conf/isc2/SymeonidisMP16}. To support sharing, \glspl{VSS} can effectively facilitate dynamic key distribution at a global scale. They can enable the occasional use of multiple types of vehicles (e.g., cars, motorbikes, scooters), catering to diverse user needs and preferences~\cite{millard2005car,le2014carsharing,FERRERO2018501}. Beyond user convenience and increased usability, by providing better utilization of available vehicles, \glspl{VSS} contribute to sustainable smart cities. This, in turn, leads to positive effects such as a reduction of emissions~\cite{DBLP:journals/tits/MartinS11}, a decrease of city congestion~\cite{shaheen2013}, and more economical use of parking space~\cite{DBLP:journals/computer/NaphadeBHPM11}.

\glspl{VSS} are gaining increased popularity: the worldwide number of users of vehicle-sharing services rose by 170\% from 2012 to 2014 (for a total of 5 million users)~\cite{acea}, while there is a tendency to reach a total of 26 million users by 2021~\cite{bert2016s}.~\footnote{Note that these predictions were made in pre-COVID-19 times.} The Car Connectivity Consortium~\cite{carconnectivity}, an organization of automotive manufactures and smartphone companies, is developing an open standard for ``smartphone-to-car'' services, where a smartphone equipped with digital keys can be used to access vehicles. The SECREDAS EU project~\cite{secredas} proposes a reference architecture for vehicular sharing~\cite{secredas:cybersecurity}, highlighting high-level security and privacy challenges that should be under consideration. The automotive supplier Valeo~\cite{valeo} in collaboration with Orange~\cite{orange} proposes an NFC based solutions for vehicular sharing~\cite{nfcw}. Volvo~\cite{volvo}, BMW~\cite{bmw}, Toyota~\cite{toyota}, Apple~\cite{apple:patent}, and several other companies have been investing in vehicle-sharing services as well. For instance, Apple announced the ``CarKey'' API in the first quarter of 2020, allowing users to (un)lock and start a vehicle using an iPhone or Apple Watch. ``CarKey'' can also be shared with other people, such as family members, enabling vehicle-sharing~\cite{apple:patent,apple:carkey}.

Despite these advantages, a major concern is the \gls{VSS} system security. An adversary may eavesdrop, and attempt to extract the key of a vehicle stored in untrusted devices, tamper with the vehicle sharing details, generate a rogue \gls{AT} to access or deny having accessed a vehicle maliciously. These significant concerns require \glspl{VSS} to deploy security mechanisms; to ensure that vehicle-sharing details cannot be tampered with by unauthorized entities, digital vehicle-keys are stored securely, and attempt to use rogue \glspl{AT} are blocked. Furthermore, it is necessary to address dispute resolution, key revocation (esp. when a user device is stolen)~\cite{gov_uk_reducing_mobile_phone_theft}, and connectivity issues~\cite{TheGuardian:vehicle-offline-connectivity,DBLP:conf/codaspy/DmitrienkoP17}. For dispute resolution, \gls{VSS} users must be accountable while their private information is protected. Current proposals to address these security issues for \gls{VSS} rely on a centralized \gls{SP}~\cite{DBLP:conf/codaspy/BusoldTWDSSS13,DBLP:conf/rfidsec/KasperKOZP13,DBLP:journals/access/WeiYWWD17,DBLP:conf/vehits/GrozaAM17,DBLP:journals/access/GrozaABMG20}, which collects all \gls{VSS} users data for every vehicle sharing-access provision, while having access to the master key of each vehicle~\cite{DBLP:conf/codaspy/DmitrienkoP17}.

However, \gls{VSS} user privacy is equally important, especially with \glspl{VSS} collecting rich personal and potentially sensitive user and vehicle data~\cite{DBLP:conf/itsc/RemeliLAB19}. An adversary may eavesdrop on data exchanges to infer sensitive information about \gls{VSS} users. For example, Enev et al.~\cite{DBLP:journals/popets/EnevTKK16} demonstrated that with 87\% to 99\% accuracy, drivers could be identified by analyzing their 15 minutes-long driving patterns. An adversary may link vehicle-sharing requests, by the same user or for the same vehicle, to deduce vehicle usage patterns and preferences; e.g., sharing patterns, such as time of use, pickup location, duration of use, person(s) a vehicle is shared with~\cite{DBLP:journals/popets/EnevTKK16}. Furthermore, the adversary could infer sensitive information about user health status by identifying vehicles for special-need passengers or their race and religious beliefs~\cite{reddit_ny_cabs}. Such user profiling would be a direct violation of the \gls{GDPR}~\cite{gdpr}. Thus, any \gls{VSS} system needs to preserve user and vehicle requests' unlinkability and keep the user and vehicle identities concealed. Furthermore, vehicle sharing operations such as \gls{AT} generation, update, or revocation should be indistinguishable in \gls{VSS}.  Towards addressing such privacy challenges, the state-of-the-art in \gls{VSS}, \exprotocol~\cite{DBLP:conf/esorics/SymeonidisAMMDP17}, proposes leveraging \gls{MPC} and focuses on privacy-preserving \gls{AT} provision, deploying multiple non-colluding servers for the generation and distribution of vehicle \glspl{AT}.

In a real-world deployment, the number of vehicles-per-user available for sharing could range from few for private individuals to thousand of vehicles for companies or (large) branches of companies~\cite{acea,statistics:num-of-cars-per-carsharing}, e.g., in car-rental scenarios~\cite{munzel2020explaining}. At the same time, the number of users registered with a \gls{VSS} can be highly varying; including all types of users, with access to varying size sets of vehicles. Designing and deploying a \gls{VSS} that serves large numbers of users and large numbers of vehicles is far from straightforward. Security and privacy safeguards significantly affect the performance of a \gls{VSS}, especially so with large numbers of users and vehicles. \exprotocol~\cite{DBLP:conf/esorics/SymeonidisAMMDP17}, although efficient ($1.55$ seconds for access provision based on an owner-single-vehicle evaluation of the protocol), has not been tested in settings replicating real-world deployment: with a large number of vehicles per user. It is paramount to have secure and privacy-preserving \glspl{VSS} that are \emph{scalable}, that is, \glspl{VSS} that remain \textit{efficient} and capable of serving users effectively as the system dimensions (number of users, number of vehicles) grow. Hence, to the best of our knowledge, there is no \gls{VSS} solution in the literature that provides security and privacy guarantees while at the same time being efficient and scalable. This work fills this gap.

\subsubsection*{Contribution} In this work, we present \system, an efficient, scalable, secure, and privacy-enhancing system for vehicle sharing-access provision that supports dispute resolution while protecting user privacy. \system\ is an extension over \exprotocol~\cite{DBLP:conf/esorics/SymeonidisAMMDP17}, and fundamentally differs in certain design choices to make it scalable and more efficient. Specifically, the contributions of this work are:
\begin{enumerate}
    \item \textit{Refined system design for improved security and privacy:} \system\ provides a comprehensive solution to vehicle sharing-access mitigating security and privacy issues considering untrusted \glspl{SP}. It deploys \acrfull{MPC} and several cryptographic primitives to ensure that \glspl{AT} are generated so that no \gls{VSS} entity other than users and vehicles learn the vehicle-sharing details. Vehicle secret keys stay oblivious towards the untrusted \gls{VSSP} although used for the \glspl{AT} generation. With the use of a \gls{PL} and anonymous communication channels~\cite{torproject} combined, \system\ also ensures the unlinkability of any two user requests, the anonymity of users, and vehicles and the indistinguishability between the \gls{AT} generation, update, and revocation operations. It also supports dispute resolution without compromising user private information while keeping users accountable. 
    
    
    \item \textit{Supporting efficiency and scalability:} We choose cryptographic primitives and underlying \gls{MPC} protocols to i) minimize the number of non-linear operations in a circuit and its circuit depth of \gls{MPC} protocols, and ii) enable parallelization of cryptographic evaluations over \gls{MPC}. For instance, optimization of \gls{MPC} consists of substituting the \gls{MAC} used in~\cite{DBLP:conf/esorics/SymeonidisAMMDP17} by an Enc-then-MAC mode. Performing \gls{MAC} operations directly on secretly shared data over \gls{MPC} is costly, as non-linear operations are the main constraint in the performance of \gls{MPC}. Instead, encrypting a message over \gls{MPC}, revealing the output, and applying \gls{MAC} to the output result in a significantly faster solution. This enables \system\ to remain efficient with multiple vehicles per user, showing a significant performance gain over \exprotocol~\cite{DBLP:conf/esorics/SymeonidisAMMDP17}. We use AES-CBC-MAC for the Boolean case and an HtMAC mode for the arithmetic case with the respective field. The latter allows parallelization, and the benchmark results in an efficient solution as HtMAC requires fewer communication rounds. These improvements are tailored towards scalable \glspl{VSS}.
    
    \item \textit{Formal semantically secure analysis:} We prove that \system\ is secure and meets its appropriate security and privacy requirements. We provide a detailed semantic security analysis overall and per security and privacy requirements extending security proofs to include the refined design and changes of the cryptographic primitives advancing \exprotocol~\cite{DBLP:conf/esorics/SymeonidisAMMDP17}.
       
    \item \textit{Improved implementation and benchmarking including a prototype OBU:} Unlike~\cite{DBLP:conf/esorics/SymeonidisAMMDP17}, we implement \system\ with the fully-fledged open-sourced \gls{MPC} framework MP-SPDZ \cite{DBLP:conf/ccs/Keller20}. The parties run an optimized virtual machine for the execution of the protocol. For comparison, we test and evaluate \system\ using two \gls{MPC} instantiations for Boolean and arithmetic circuits. Our performance evaluation demonstrates that \system\ can be highly efficient even for users with thousand of vehicles, hence making it ready for real-world deployment. Its significant improvement shows that it requires only \boldmath{$\approx \stepBmacCBCtime$}~ms for a owner-single-vehicle \gls{AT} generation (\boldmath{$\stepBCBCmoreAT$} times faster compared to~\cite{DBLP:conf/esorics/SymeonidisAMMDP17}). Simultaneously, it can handle multiple \gls{AT} generations per second (\boldmath{$\approx \stepBMiMCmultiVehATs$}~\glspl{AT}/s) for owner-multi-vehicles individuals and (branches of) rental companies, resulting in an efficient and scalable solution that is ready for real-world deployment. Furthermore, we implement the \gls{AT} verification on a prototype \gls{OBU}, demonstrating that \system\ is practical on the vehicle side too.
\end{enumerate}


The rest of the paper is organized as follows: Section~\ref{sec:system_model} provides the system model and preliminaries on \system. Section~\ref{sec:crypto_build_blocks} describes the cryptographic building blocks used in \system. Section~\ref{sec:system} describes the system in detail and Section~\ref{sec:extended_analysis} provides the security and privacy analysis of \system. Section~\ref{sec:protocol_evaluation} evaluates its performance and complexity, and demonstrates its efficiency and scalability. Section~\ref{sec:related_work} gives an overview of the state-of-the-art related work. Section~\ref{sec:conclusion} concludes our work.



\begin{table*}[]
\centering
\caption{Notation.}
\label{table:notations}
\resizebox{\textwidth}{!}{%
    \begin{tabular}{l}
        \raisebox{-\totalheight}{\includegraphics{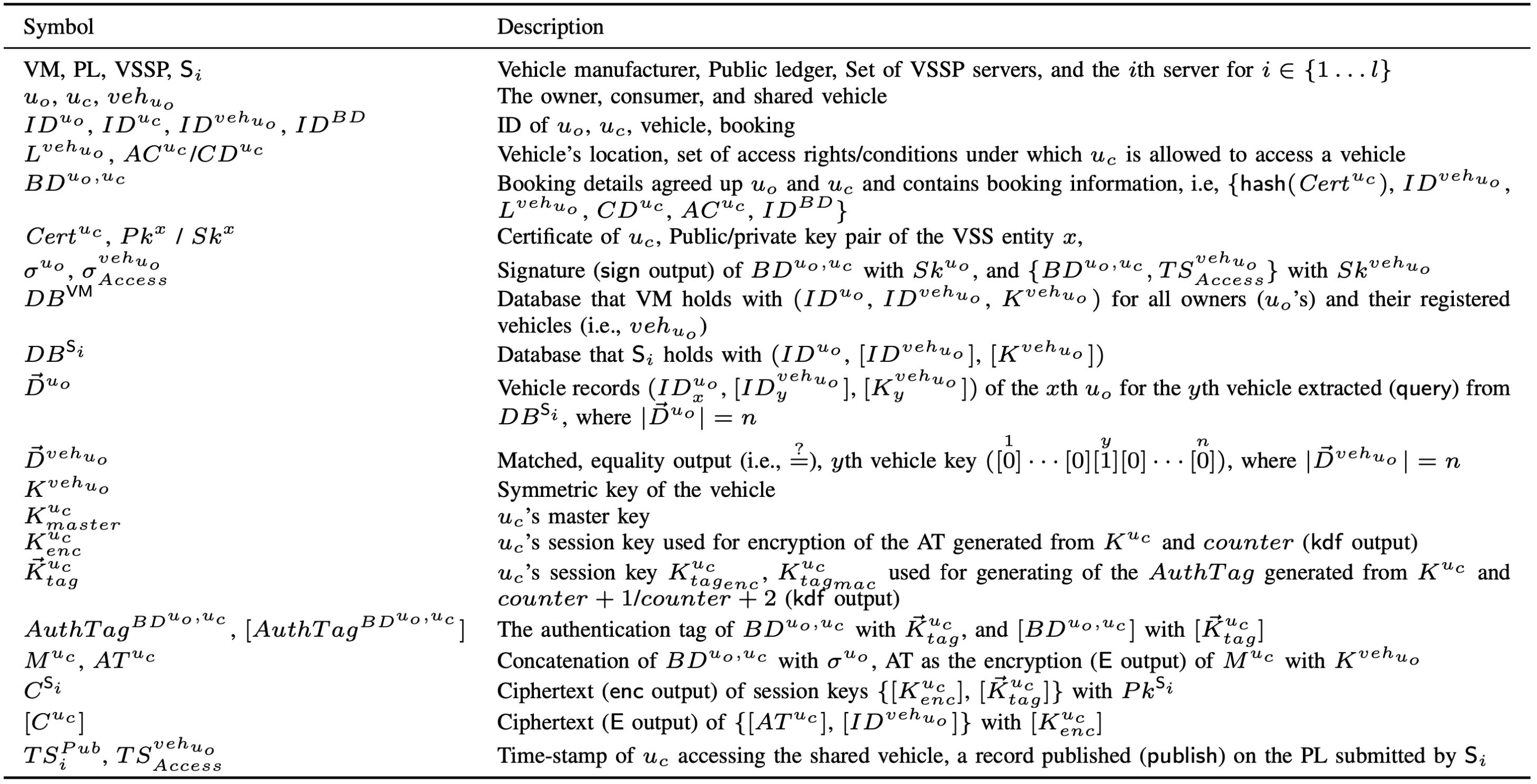}}
    \end{tabular}%
}
\end{table*}

\section{System, Adversarial Models and Requirements}\label{sec:system_model}
We outline a system model of \glspl{VSS}, along with the adversarial model. We provide the functional, security, privacy, and performance requirements any secure and privacy-enhancing \gls{VSS} needs to satisfy.
    
    \begin{figure}[t]
        \centering
            \resizebox{\columnwidth}{!}{%
            \includegraphics{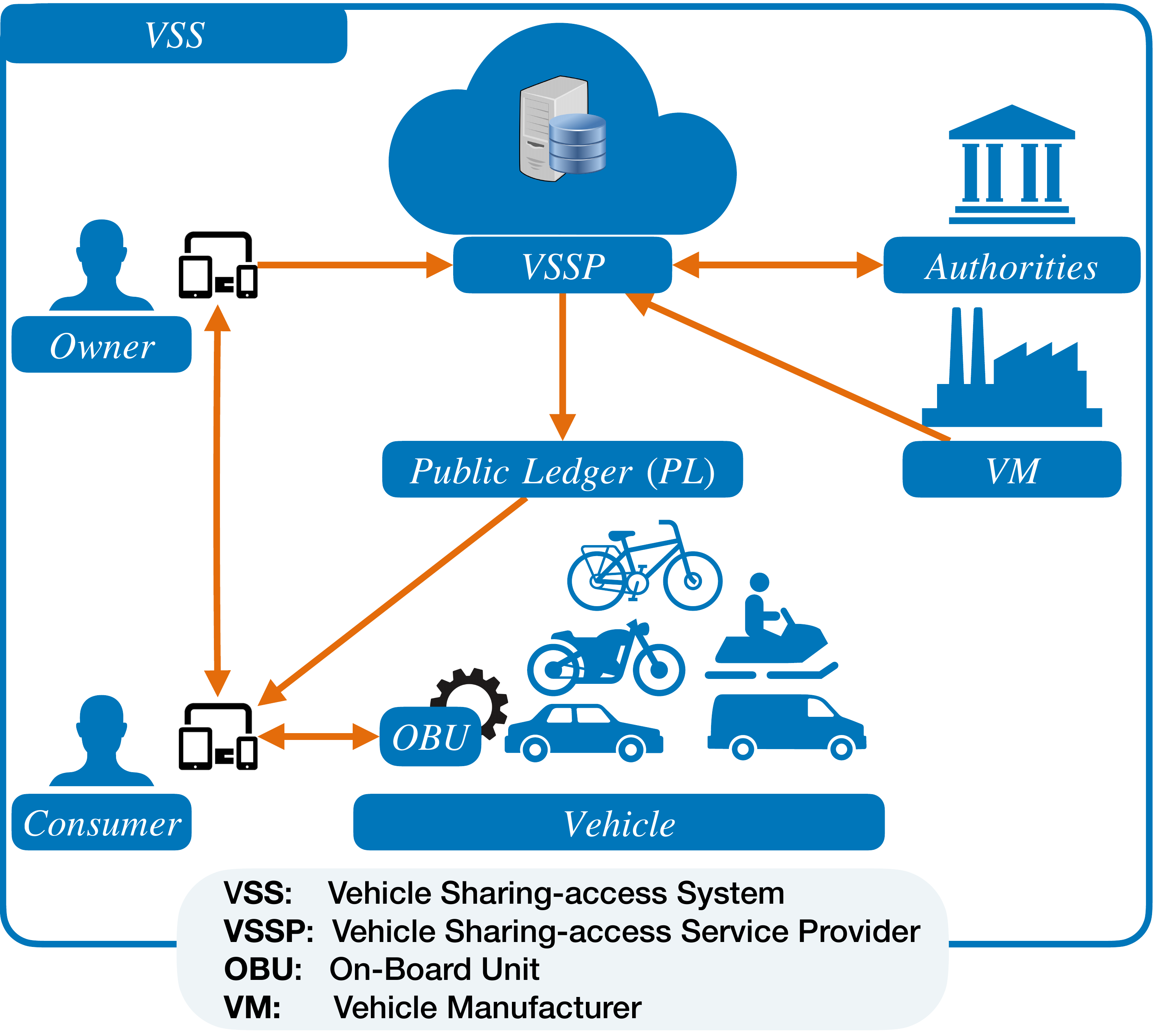}}
            \caption{\acrfull{VSS} model. } 
        \label{VSS}
    \end{figure}
    
    \subsection{System model} A \gls{VSS} is comprising of users, vehicles, vehicle-manufacturers, and authorities as it is illustrated in Fig.~\ref{VSS}. We consider two types of users: \textit{owners} ($u_o$), individuals or vehicle rental companies willing to share (rent out) their vehicles, and \textit{consumers} ($u_c$), individuals using vehicles available for sharing; both use \acrfullpl{PD}, such as smartphones, to interact with \gls{VSS} entities and each other. The \gls{OBU} is a hardware/software component that enables vehicle connectivity~\cite{nexcom:vtc6201-ft,preserve,ivners}. It is equipped with a secure wireless interface for short-range (e.g., NFC, or Bluetooth) or over the Internet interface (e.g., cellular data) for accessing securely the \textit{vehicle}. The \gls{VM} is responsible for managing the digital keys that enable access into each vehicle. These keys are used for enabling vehicle sharing-access in \gls{VSS} as well. The \gls{VSSP} is a cloud infrastructure that facilitates the vehicle \gls{AT} generation, distribution, update, and revocation. It consists of \textit{servers} that collaboratively generate \glspl{AT} and publish them on the \gls{PL}, a secure public bulletin board~\cite{micali2016algorand}. Prior to each vehicle sharing-access session, the owner and consumer agree on the booking details. We denote $BD^{u_o, u_c}$, $AT^{veh_{u_o}}$, and $K^{veh_{u_o}}$ as \acrfull{BD}, an \acrfull{AT} for a vehicle, and for the vehicle secret key, respectively.

    \subsection{Assumptions} We assume the existence of secure and authenticated communication over all channels between entities at \gls{VSS}, e.g., by using SSL-TLS~\cite{rfc8446} or NFC. 
   There is a \gls{PKI} in place (e.g.,~\cite{DBLP:journals/tits/KhodaeiJP18}), and each entity has a digital certificate and a corresponding private/public-key pair. The \gls{VSSP} servers are managed by organizations with conflicting interests, such as user unions, \glspl{VM}, and authorities. Thus, these are non-colliding organizations. The intra-\gls{VSSP} communication is considered to be up to a 10Gb/s network. The \gls{OBU} has an embedded \gls{HSM}~\cite{automotive_tpm} that enables secure execution environment and key storage. Before each evaluation, the \gls{BD} are agreed upon by the owner and consumer. Both keep the \gls{BD} confidential against external parties. The \gls{BD} contains the owner, consumer, and vehicle identities and the location and time duration of the reservation.

   \subsection{Adversarial Model} The \gls{VSSP}, the \gls{PL}, and the \gls{VM} are passive adversaries, i.e., honest-but-curious in our case. They execute the protocol correctly, but they may attempt to deduce user private information. The owners can be passive adversaries, as they hold information about the booking, but they will not deviate from the protocol. Consumers and outsiders can be active adversaries aiming to illegally access a vehicle, alter the booking information, and hide incidents. Authorities are trusted entities only for specific transactions in case of disputes. The vehicle, specifically its \gls{OBU}, is trusted and tamper-evident designed to resist accidental or deliberate physical destruction (i.e., it serves as an event data recorder equipped with software and hardware security mechanisms~\cite{automotive_tpm}). User \glspl{PD} are untrusted as they can get stolen, broken, or lost. Relay attacks, which can be tackled with distance bounding protocols~\cite{DBLP:conf/eurocrypt/BrandsC93}, are left out of the scope of this paper.

    \subsection{System Design Requirements} We detail functional, security, privacy, and performance requirements that a \gls{VSS} should satisfy, denoted \emph{FR}, \emph{SR}, \emph{PR}, and \emph{ESR}, respectively. The list builds on the requirements specified in~\cite{DBLP:conf/isc2/SymeonidisMP16}, extending the ones of \exprotocol~\cite{DBLP:conf/esorics/SymeonidisAMMDP17}.
    
    \textit{Functional requirements:}
    
    \begin{itemize}
    
    \item \textit{FR\fr -- Offline vehicle access.} Vehicle access should be supported in locations with no (or limited) network connectivity.
    
    \item \textit{FR\fr --  \acrfull{AT} update and revocation by the owner $u_o$.} No-one except the owner, $u_o$, can initiate an \gls{AT} update or revocation.
    
    \end{itemize}
    
    \textit{Security requirements:}
    
    \begin{itemize}
    \item \textit{SR\sr -- Confidentiality of \acrfull{BD}, $BD^{u_o, u_c}$.} No-one except the owner $u_o$, consumer $u_c$, and the shared vehicle $veh_{u_o}$ should access $BD^{u_o, u_c}$.
    
    \item \textit{SR\sr -- Entity and data authenticity of $BD^{u_o, u_c}$ from the owner $u_o$.} The origin and integrity of the \gls{BD}, $BD^{u_o, u_c}$, by the owner $u_o$ should be verified by the shared vehicle $veh_{u_o}$.
    
    \item \textit{SR\sr -- Confidentiality of $AT^{veh_{u_o}}$.} No-one except the consumer $u_c$ and the shared vehicle $veh_{u_o}$ should access $AT^{veh_{u_o}}$.
    
    \item \textit{SR\sr -- Confidentiality of vehicle key, $K^{veh_{u_o}}$.} No-one except the \gls{VM} and the shared vehicle $veh_{u_o}$ should hold a copy of vehicle's key $K^{veh_{u_o}}$. 
    
    \item \textit{SR\sr -- Backward and forward secrecy of $AT^{veh_{u_o}}$.} Compromise of a session key used to encrypt any $AT^{veh_{u_o}}$ should not compromise future and past \glspl{AT} published on \gls{VSS} e.g., on the \gls{PL}, for any honest consumer $u_c$.
    
    \item \textit{SR\sr -- Non-repudiation of origin of $AT^{veh_{u_o}}$.} The owner $u_o$ should not be able to deny agreeing on \gls{BD} terms, $BD^{u_o, u_c}$, or deny initiating the corresponding \gls{AT} generation operation for $AT^{veh_{u_o}}$.
    
    \item \textit{SR\sr -- Non-repudiation of $AT^{veh_{u_o}}$ receipt by $veh_{u_o}$ at $u_o$.} The consumer $u_c$ should not be able to deny receiving and using the $AT^{veh_{u_o}}$ to open and access the $veh_{u_o}$ (once it has done so).
    
    \item \textit{SR\sr -- Accountability of users (i.e., owner $u_o$ and consumer $u_c$).} On a request of law enforcement, \gls{VSSP} should be able to supply authorities with the vehicle-access transaction details without compromising the privacy of other users.
    
    \end{itemize}
    
    \textit{Privacy requirements:}
    
    \begin{itemize}
    
    \item \textit{PR\pr -- Unlinkability of (any two) requests of any consumer, $u_c$, and the vehicle, $veh_{u_o}$(s).} No-one except the onwer $u_o$, the consumer $u_c$, and the shared vehcile $veh_{u_o}$ should be able to link two booking requests of any consumer $u_c$ and for any shared vehicle $veh_{u_o}$ linking their identities, i.e.,  $ID^{u_c}$, and $ID^{veh_{u_o}}$.
    
    \item \textit{PR\pr -- Anonymity of any consumer, $u_c$, and vehicle, $veh_{u_o}$.} No-one except the owner $u_o$, the consumer $u_c$, and the shared vehicle $veh_{u_o}$ should learn the identity of $u_c$ and $veh_{u_o}$.
    
    \item \textit{PR\pr -- Indistinguishability of $AT^{veh_{u_o}}$ operations.} No-one except the owner $u_o$, the consumer $u_c$ and the vehicle $veh_{u_o}$, should be able to distinguish between operations of generation, update and revocation of the \gls{AT}, $AT^{veh_{u_o}}$.
    
    \end{itemize}
    
    \textit{Performance requirement:}
    \begin{itemize}
        \item \textit{ESR\psr -- Efficiency and scalability in a real-world deployment.} The \gls{VSS} should remain capable of efficiently and effectively servicing users, as their numbers and the numbers of vehicles per user increase to levels required for real-world deployment.
    \end{itemize}

\section{Cryptographic Building Blocks}\label{sec:crypto_build_blocks}

\begin{figure*}[ht]
    \centering
    \resizebox{\textwidth}{!}{%
\includegraphics{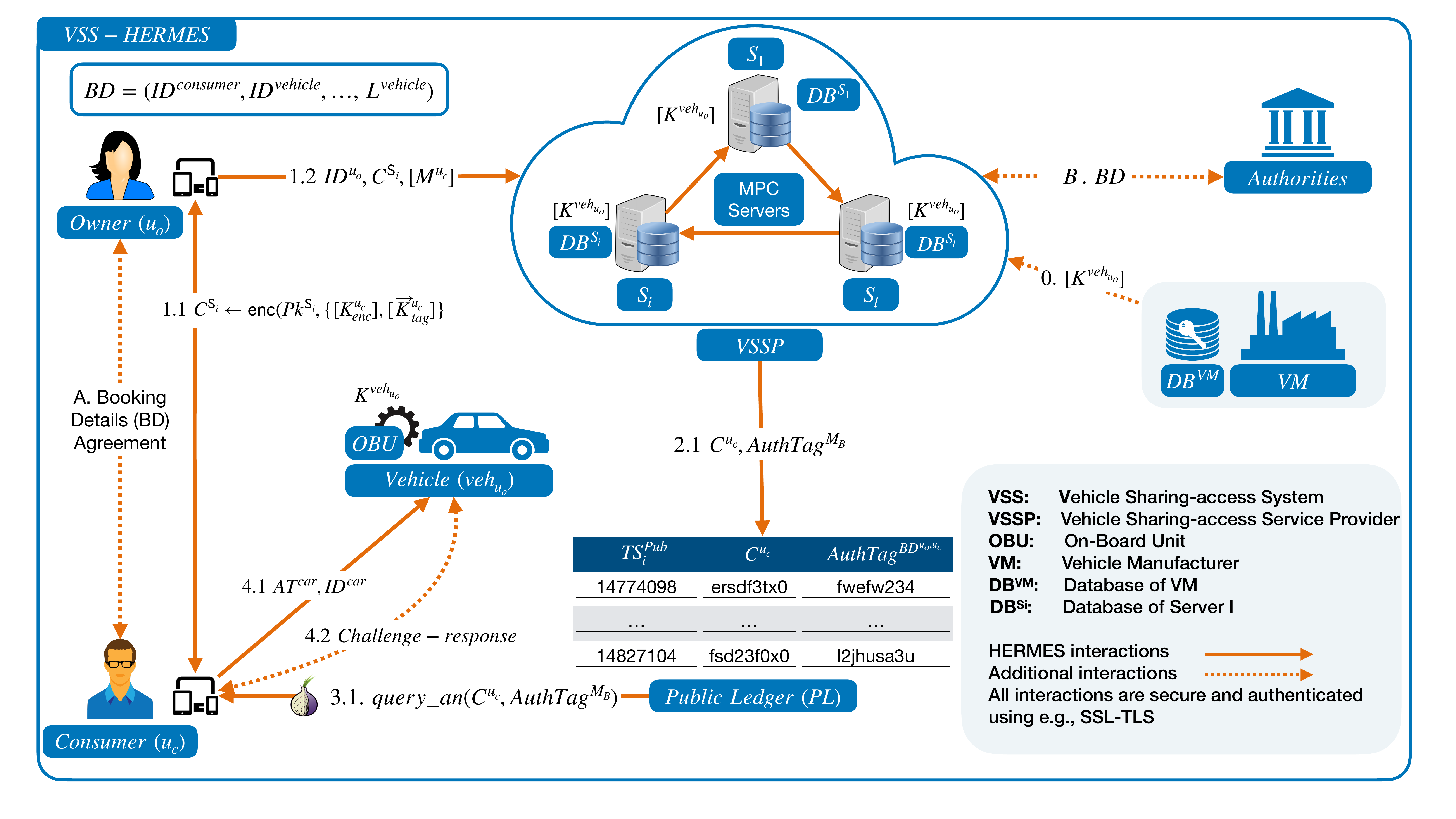}
}
    \caption{\system\ high level overview. Numbers correspond to the steps outlined in the text of Section~\ref{sec:system}. Figures \ref{fig:step1}, \ref{fig:step2}, \ref{fig:step3} and \ref{fig:step4} describe Steps~1, 2, 3 and 4 in more detail.}
    \label{fig:vss_overview}
\end{figure*}

\subsubsection{Cryptographic Primitives}
\system\ uses cryptographic building blocks, as described below. For each of the building blocks, we provide concrete instantiations we use in our proof-of-concept implementation detailed in Section~\ref{sec:protocol_evaluation}.
\begin{itemize}
    \item \textit{Signagure scheme:} $\sigma \leftarrow \mathsf{sign}(Sk,m)$ and $\mathsf{true}/\mathsf{false} \leftarrow \mathsf{verify}(Pk, m,\sigma)$ are public-key operations for signing and verification respectively. These can be implemented using RSA, as defined in the PKCS $\#1$ v2.0 specification~\cite{rfc_2437}.
    \item \textit{Key derivation function:} $K \leftarrow \mathsf{kdf}(K, counter)$ is a key derivation function using a master key and a counter as inputs. It can be based on a \gls{PRF} and implemented using CTR mode with AES~\cite{barker2012nist}.~\footnote{In our case, the message input is small, i.e., $\ll 2^{64}$ blocks for AES in CTR, and the generation is performed with side channel attacks not to be a concern~\cite{DBLP:conf/sp/CohneyKPGHRY20}.}
    \item \textit{Public key encryption/decryption:} $c \leftarrow \mathsf{enc}(Pk,m)$ and $m \leftarrow \mathsf{dec}(Sk,c)$ are encryption and decryption functions based on public key primitives. These can be implemented using RSA, as defined in the RSA-KEM specifications~\cite{rfc_5990}.
    \item \textit{Symmetric key encryption/decryption:} $c \leftarrow \mathsf{E}(K,m)$ and $m \leftarrow \mathsf{D}(K,c)$ are encryption and decryption functions based on symmetric key primitives. These can be implemented using AES in CTR mode.
    \item \textit{Cryptographic hash:} $z \leftarrow \mathsf{hash}(m)$ it is a message digest function. This can be SHA-2 or SHA-3.
    \item \textit{Message Authentication Code:} $t \leftarrow \mathsf{mac}(k, m)$ is a cryptographic \gls{MAC} that outputs an authentication tag, $t$, given a message $m$ and a key $k$. These can be implemented using CBC-MAC-AES or HtMAC-MiMC.
\end{itemize}
Furthermore, we use $z \leftarrow \mathsf{query}(x,y)$ to denote the retrieval of the $x$th value from the $y$th database $DB$ (to be defined in Sect.~\ref{sec:system}), and $z \leftarrow \mathsf{query\_an}(y)$ to denote the retrieval of the $y$th value from the \gls{PL} through an anonymous communication channel such as Tor~\cite{torproject}, aiming to anonymously retrieve a published record (e.g., \gls{AT}) submitted using the $\mathsf{publish}(y)$ function.

\subsubsection{Multiparty Computation}
MPC allows a set of parties to compute a function over their inputs without
revealing them. To evaluate a function on secret inputs using \gls{MPC}, one needs to unroll the function to a series of additions and multiplications in a field. Following the seminal papers of Yao for the two-party
case~\cite{DBLP:conf/focs/Yao86} and by Goldreich, Micali and Wigderson in the
multiple parties setting~\cite{micali1987play}, secure \gls{MPC} has gained
much traction in the past years with many open-source
frameworks~\cite{hastings2019sok}.


Our algorithms use building blocks whose instantiation depends on the protocol type. However, they can be treated generically.
This is also called an arithmetic black-box functionality~\cite{damgaard2003universally}. The functionality mainly in use consists of:
\begin{itemize}
    \item \textit{Secret sharing function:} $[x] \leftarrow \mathsf{share}(x)$ is a function that inputs $x$ and outputs $[x]$ in secret shared form to all parties. The underlying secret sharing scheme is described in Araki et al.~\cite{DBLP:conf/ccs/ArakiFLNO16}.
    \item \textit{Shares reconstruction:} $x \leftarrow \mathsf{open}([x])$ which takes a secret shared value $[x]$ and opens it, making $x$ known to all parties.
    \item \textit{Equality check:} $[z] \leftarrow ([x] \stackrel{?}{=} [y])$ outputs a secret bit $[z]$ where $z \in \{0, 1\}$. If $x$ is equal to $y$ then set $z \leftarrow 1$ otherwise set $z \leftarrow 0$. Note that for the large field case there is a statistical security parameter $\mathsf{sec}$, whereas for the $\mathbb{F}_2$ case the comparison is done with perfect security (i.e. no $\mathsf{sec}$ parameter). The equality operator is implemented using the latest protocols of Escudero et al.~\cite{DBLP:conf/crypto/0001GKRS20}.
    
    \item $c \leftarrow \mathsf{E}([k], [m])$ An encryption function, i.e., $\mathsf{E}$, takes as inputs a secret shared key $[K]$ and a vector of $128$ bit blocks $[m]$. For the $\mathbb{F}_2$ case, $\mathsf{E}$ is implemented using AES in CTR mode. Concretely, the AES circuit description is the one from SCALE-MAMBA~\cite{aly2020scale}, which has $6400$ AND gates. For the $\mathbb{F}_p$ case, MiMC is used as a \gls{PRF} in counter mode as presented in~\cite{rotaru2017modes} to take advantage of \gls{PRF} invocations done in parallel.
    \item $t \leftarrow \mathsf{mac}([k], [m])$ is a
    tag generation function for secret shared key $[k]$ and message $[m]$.
    For the case when inputs are in a large field, we will not compute the \gls{MAC} as above, but rather as $\mathsf{mac}([k], \mathsf{E}([k'], [m]))$. The reason is that, according to~\cite{rotaru2017modes}, we can obtain a more efficient cryptographic \gls{MAC} in MPC by first computing $\mathsf{E}([k'], [m])$ in parallel with a secret shared key $[k']$, opening the result, and evaluate the \gls{MAC} function in the clear.
    Their optimizations hold only for arithmetic
    circuits with HtMAC over a large field~\cite{chida2018fast}, although they could likely be extended to Boolean circuits as well.
    In the Boolean case,
    the $\mathsf{mac}$ function is implemented as CBC-MAC-AES. Note that for the
    $\mathbb{F}_2$ case there are more efficient ways to do this, but we keep
    CBC-MAC as a comparison baseline to \exprotocol~\cite{DBLP:conf/esorics/SymeonidisAMMDP17}.

\end{itemize}

\section{\system}\label{sec:system}
In this section, we present \system\ in detail. We provide the complete system description; its entities, the functional and cryptographic operations performed and messages exchanged (see Fig.~\ref{fig:prot} in Appendix~\ref{appendix:protocol_complete}). Prior to explaining \system\ in detail, we provide a brief description overview as an introduction.

\subsection{Overview of \system}
We consider a single owner, a single consumer, and a set of shared vehicles for simplicity in presentation, without loss of generality. There are two prerequisite steps: \textit{vehicle key distribution (Step~A)} and establishing the details for the \textit{vehicle booking (Step~B)}. In a nutshell, as vehicle owners register their vehicles, the \gls{VSSP}, using the owner identity, retrieves the vehicle identity and the corresponding key from \gls{VM} in Step~A. Note that the \gls{VM}, a trusted \gls{SP} for \gls{VSS}, holds all the secret keys of vehicles. Both the identity and the vehicle key are transferred from \gls{VM} to \gls{VSSP} in a secret-shared form~\cite{DBLP:conf/ccs/ArakiFLNO16}, that is indistinguishable from randomness~\cite{DBLP:journals/cacm/Shamir79}. Thus, there is nothing the \gls{VSSP} can deduce from the vehicle identity and corresponding digital master key. For each initialization of \system, the \gls{BD} was specified between the owner and the consumer, tailored to each vehicle sharing agreement. During \textit{vehicle booking} in Step~B, the owner and consumer specified the identity of the vehicle from the pool of vehicles, the duration of the reservation, the access rights, and location.~\footnote{Note that \system\ is agnostic to the specificities of \gls{BD} drawing from the analogy in \gls{VSS} from car-rental scenarios.}


\system\ consists of four main steps: \textit{session key generation and \gls{BD} distribution} (Step~1), \textit{\acrfull{AT} generation} (Step~2), \textit{\acrfull{AT} distribution and verification} (Step~3), and \textit{vehicle access} (Step~4). During the \textit{session key generation and data distribution} in Step~1, the consumer generates three session keys. One of these session keys is used to \textit{encrypt} the generated \gls{AT} at the \gls{VSSP} servers, so that only the consumer has access to it. The two other session keys are used to generate an \textit{authentication tag} of the \gls{BD}, such that only the consumer can identify and retrieve the \gls{AT} from the \gls{PL} as well as verify that the beforehand agreed \gls{BD} is included in the \gls{AT}. As the consumer considers the owner and the \gls{VSSP} as honest-but-curious entities, the consumer conceals the three-session keys before forwarding them -- the keys are transformed in secret shared form~\cite{DBLP:conf/ccs/ArakiFLNO16}. Moreover, to protect its identity, the consumer avoids direct communication with \gls{VSSP} by forwarding the shares of session keys to the owner. The owner then forwards to the \gls{VSSP} the \gls{BD} and its signature in a shared form, together with the concealed session keys, to each $\mathsf{S}_i$ server of \gls{VSSP}. Once each $\mathsf{S}_i$ of \gls{VSSP} receives the shares of the session keys and the booking details, the \textit{\acrfull{AT} generation}, Step~2, commences. The vehicle key is retrieved from the \gls{DB} in each $\mathsf{S}_i$ server, using an equality check over \gls{MPC}, thus preserving the key secrecy. The \gls{AT} is generated by encrypting the \gls{BD} and its signature with the vehicle key, such that only the vehicle itself can retrieve them. Moreover, the session keys, generated by the consumer, are used to encrypt the \gls{AT}, and also create an \textit{authentication tag}, such that only the consumer can identify and access the \gls{AT}. Each of the servers, $\mathsf{S}_i$, then forwards the encrypted \gls{AT} and its authentication tag to the \gls{PL}. The \gls{PL} serves as a bulletin board and notifies the \gls{VSSP} once it publishes the information. At the \textit{\acrfull{AT} distribution and verification}, Step~3, the consumer can identify and retrieve the corresponding \gls{AT}. As the consumer considers the \gls{PL} as honest-but-curious, it hides its identity (i.e., IP address) by querying the \gls{PL} using an anonymous communication channel such as Tor~\cite{torproject}. The consumer then retrieves the \gls{AT}, to be used by the vehicle, to verify and allow access to the consumer for the predefined booking duration of \textit{vehicle access} at Step~4.

\subsection{\system\ in Detail}
We first describe the \textit{prerequisite} steps. We detail the core operations in four steps. Table~\ref{table:notations} lists the notation used throughout the paper. 


\subsubsection*{Prerequisite steps} Before \system\ commences, two prerequisite steps are necessary: \textit{vehicle key distribution} and establishing the details for booking, i.e., \textit{vehicle booking}. 

\paragraph*{Step A - Vehicle key distribution} This step takes place immediately after the $x$th owner, $ID^{u_o}_x$, registers her $y$th vehicle, $ID^{veh_{u_o}}_y$, with the \gls{VSSP}. The \gls{VSSP} request from the \gls{DB} of \gls{VM}, $DB^{VM}$, the secret symmetric key of the vehicle, $K^{veh_{u_o}}_y$, and the corresponding identity of the owner, $ID^{veh_{u_o}}_y$, i.e.,

\begin{equation*}
    DB^{VM} = 
    \begin{pmatrix}
        ID^{u_o}_1 & ID^{veh_{u_o}}_1 & K^{veh_{u_o}}_1 \\
        \vdots & \vdots & \vdots \\
        ID^{u_o}_x& ID^{veh_{u_o}}_y & K^{veh_{u_o}}_y \\
        \vdots & \vdots & \vdots \\
        ID^{u_o}_m& ID^{veh_{u_o}}_n & K^{veh_{u_o}}_n \\
    \end{pmatrix}.
\end{equation*}

Then, \gls{VM} replies and \gls{VSSP} retrieves these values in secret shared form, denoted by $[K^{veh_{u_o}}_y]$ and $[ID^{veh_{u_o}}_y]$, respectively. It stores, $ID^{u_o}_x$, $[ID^{veh_{u_o}}_y]$ and $[K^{veh_{u_o}}_y]$ in its \gls{DB} denoted $DB^{\mathsf{S}_i}$, i.e.,

\begin{equation*}
   DB^{\mathsf{S}_i} = 
    \begin{pmatrix}
        ID^{u_o}_1 & [ID^{veh_{u_o}}_1] & [K^{veh_{u_o}}_1] \\
        \vdots & \vdots & \vdots \\
        ID^{u_o}_x& [ID^{veh_{u_o}}_y] & [K^{veh_{u_o}}_y] \\
        \vdots & \vdots & \vdots \\
        ID^{u_o}_m& [ID^{veh_{u_o}}_n] & [K^{veh_{u_o}}_n] \\
    \end{pmatrix}.
\end{equation*}

For simplicity, we use the $ID^{u_o}$, $ID^{veh_{u_o}}$ and $K^{veh_{u_o}}$ instead of $ID^{u_o}_x$, $ID^{veh_{u_o}}_y$ and $K^{veh_{u_o}}_y$ throughout the paper.

\paragraph*{Step B - Vehicle booking} This step allows the owner and consumer to agree on the \gls{BD} before \system\ commences. In specific, $u_o$ and $u_c$ to agree on the booking details, i.e., $ BD^{u_o, u_c} = \{\mathsf{hash}(\mathit{Cert}^{u_c})$, $ID^{veh_{u_o}}$, $L^{veh_{u_o}}$, $CD^{u_c}$, $AC^{u_c}$, $ID^{BD}\}$, where $\mathsf{hash}(\mathit{Cert}^{u_c})$ is the hash value of the digital certificate of $u_c$, $L^{veh_{u_o}}$ is the pick-up location of the vehicle, $CD^{u_c}$ is the set of conditions under which $u_c$ is allowed to use the vehicle (e.g., restrictions on locations, time period), $AC^{u_c}$ are the access control rights based on which $u_c$ is allowed to access the vehicle, and $ID^{BD}$ is the booking identifier. 




\subsubsection*{\system\ operations in four steps}


\begin{figure*}[hbt!]
    \centering
        \resizebox{0.95\textwidth}{!}{%
        \includegraphics{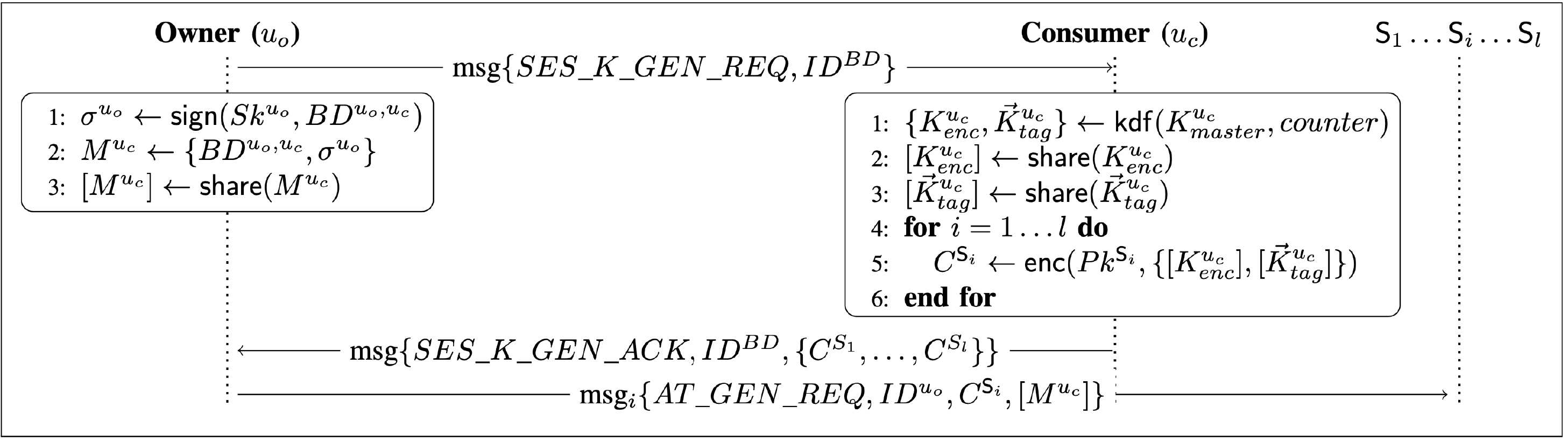}}
        \caption{Step 1: session key generation and \gls{BD} distribution.} 
    \label{fig:step1}
\end{figure*}

\paragraph*{Step 1 -- Session key generation and \gls{BD} distribution}\label{step1}
While $u_o$ signs the booking details, $BD^{u_o, u_c}$, $u_c$ generates session keys for encryption and data authentication, i.e., $K^{u_c}_{enc}$ and $\vec{K}^{u_c}_{tag}=(K^{u_c}_{tag_{mac}},K^{u_c}_{tag_{enc}})$, respectively. The generated material by $u_c$ and $u_o$ are sent to each $\mathsf{S}_i$ via $u_o$. These will be used for the generation of the \gls{AT}.

In detail, as depicted in Fig.~\ref{fig:step1}, $u_o$ sends a request for \textit{session-key-generation}, \textit{SES\_K\_GEN\_REQ}, together with $ID^{BD}$ to $u_c$. Once it receives the request, $u_c$ generates the session keys, $K^{u_c}_{enc}$ and $\vec{K}^{u_c}_{tag}$. $K^{u_c}_{enc}$ is used by the \gls{VSSP} servers, $\mathsf{S}_i$, to encrypt the \gls{AT}, and ensure that only $u_c$ has access to it. Note that each $\mathsf{S}_i$ does encryption evaluations in a secret shared way. $\vec{K}^{u_c}_{tag}$ is used to generate an authentication tag, allowing $u_c$ to verify that \gls{AT} contains $BD^{u_o, u_c}$ agreed upon during the \textit{vehicle booking}. It utilizes a $\mathsf{kdf}()$ function with $u_c$'s master key as an input, i.e., $K^{u_c}_{master}$ and a $counter$. For $\vec{K}^{u_c}_{tag}$, two session keys are generated and stored: one for encryption, $K^{u_c}_{tag_{enc}}$ (i.e., $\vec{K}^{u_c}_{tag}[0] = K^{u_c}_{tag_{enc}}$), and one for authentication, $K^{u_c}_{tag_{mac}}$ (i.e., $\vec{K}^{u_c}_{tag}[1] = K^{u_c}_{tag_{mac}}$). Then, $u_c$ constructs $\ell$ secret shares of $[K^{u_c}_{enc}]$ and $[\vec{K}^{u_c}_{tag}]$, one for each $\mathsf{S}_i$. This ensures that none of the servers alone has access to these session keys. Nonetheless, they can jointly perform evaluations utilizing the shares of these keys.



The consumer encrypts $[K^{u_c}_{enc}]$ and $[\vec{K}^{u_c}_{tag}]$ with the public-key of each $\mathsf{S}_i$, $C^{\mathsf{S}_i} = \mathsf{enc}(Pk^{\mathsf{S}_i},\{[K^{u_c}_{enc}], [\vec{K}^{u_c}_{tag}]\})$. It ensures that only the specific $\mathsf{S}_i$ can access the corresponding shares. Finally, $u_c$ forwards to $u_o$ an acknowledgment message, \textit{SES\_K\_GEN\_ACK}, along with $ID^{BD}$ and $\{C^{\mathsf{S}_1}, \dots, C^{\mathsf{S}_l}\}$. The owner, $u_o$, signs $BD^{u_o, u_c}$ with her private key, i.e., $\sigma^{u_o} = \mathsf{sign}(Sk^{u_o},BD^{u_o, u_c})$. In a later stage, the vehicle will use $\sigma^{u_o}$ to verify that $BD^{u_o, u_c}$ was approved by $u_o$. Then $u_o$ transforms $M^{u_c}=\{BD^{u_o, u_c},\sigma^{u_o}\}$ into $\ell$ secret shares, i.e., $[M^{u_c}]$. Upon receipt of the response of $u_c$, $u_o$ forwards to each $\mathsf{S}_i$ an access-token-generation request, \textit{AT\_GEN\_REQ}, along with $ID^{u_o}$, the corresponding $C^{\mathsf{S}_i}$ and $[M^{u_c}]$.



\begin{figure*}[hbt!]
    \centering
        \resizebox{0.95\textwidth}{!}{%
        \includegraphics{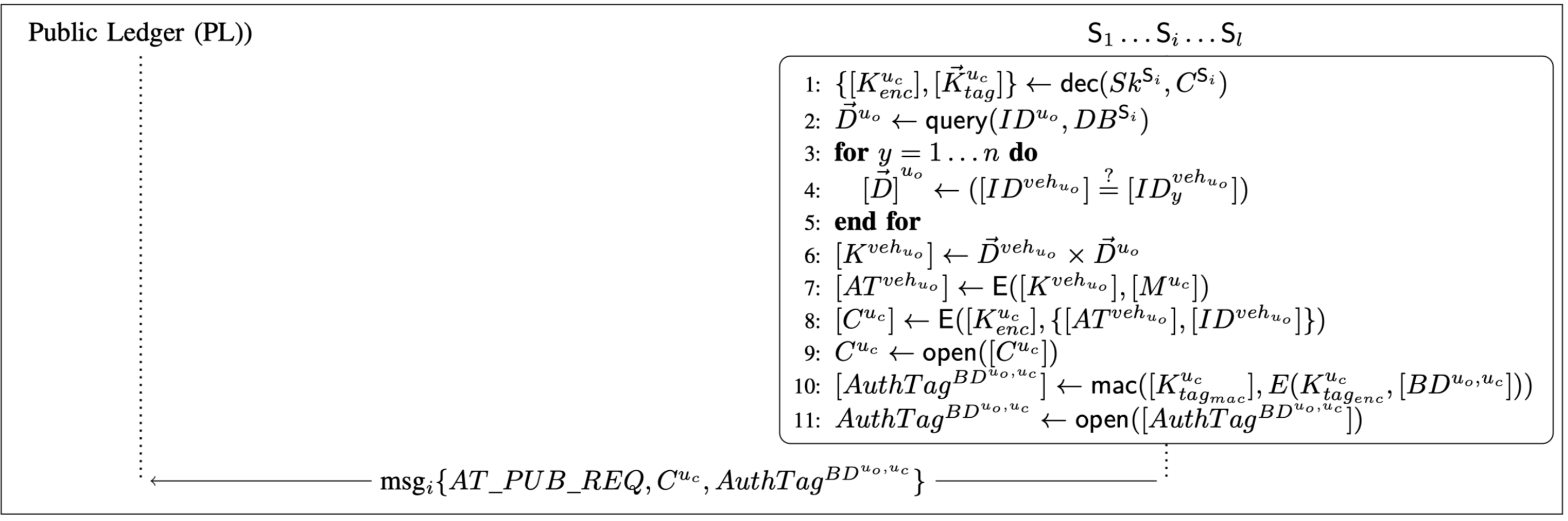}}
        \caption{Step 2: \gls{AT} generation.} 
    \label{fig:step2}
\end{figure*}

\paragraph*{Step 2 -- \acrfull{AT} generation}\label{step2}
The servers generate an \gls{AT} and publish it on the \acrfull{PL}. 

In detail, as depicted in Fig.~\ref{fig:step2}, after receiving the \textit{AT\_GEN\_REQ} from $u_o$, the servers obtain the session key shares, $\{[K^{u_c}_{enc}], [\vec{K}^{u_c}_{tag}]\}$. Each $\mathsf{S}_i$ decrypts $C^{\mathsf{S}_i}$ using its private key. Session keys are for encrypting the \gls{AT} used to access a vehicle by $u_c$ and for generating an authentication tag used by $u_c$ to verify the data authenticity of \gls{BD} contained in the \gls{AT}, respectively.

To generate the \gls{AT}, $[AT^{veh_{u_o}}]$, the key of the vehicle, $[K^{veh_{u_o}}]$, is retrieved from $DB^{\mathsf{S}_i}$ using query and equality check operations as proposed in~\cite{DBLP:conf/esorics/SymeonidisAMMDP17}. In specific, for each $\mathsf{S}_i$, it uses the $ID^{u_o}$ to extract $[K^{veh_{u_o}}]$ from $DB^{\mathsf{S}_i}$. The result is stored in a vector $\vec{D}^{u_o}$ of size $n\times3$, i.e.,
\begin{equation*}
    \vec{D}^{u_o} = 
        \begin{pmatrix}
ID^{u_o} & [ID^{veh_{u_o}}_1] & [K^{veh_{u_o}}_1] \\
\vdots & \vdots & \vdots \\
ID^{u_o} & [ID^{veh_{u_o}}_y] & [K^{veh_{u_o}}_y] \\
\vdots & \vdots & \vdots \\
ID^{u_o} & [ID^{veh_{u_o}}_{n}] & [K^{veh_{u_o}}_{n}]
        \end{pmatrix},
\end{equation*}
where $n$ is the number of vehicles owned by $u_o$ and registered with the \gls{VSS}.

To retrieve the record for the vehicle to be shared, each $\mathsf{S}_i$ uses the $([x] \stackrel{?}{=} [y])$ operation to extract $[ID^{veh_{u_o}}]$ from $[M^{u_c}]$ performing an equality check with each of the $n$ records of $\vec{D}^{u_o}$. The comparison outputs $1$ for identifying the vehicle at position $y$ or $0$ in case of mismatch. The results are stored in a vector $\vec{D}^{veh_{u_o}}$ of length $n$, i.e.,
    \begin{equation*}
        \vec{D}^{veh_{u_o}}=
        \Big(\overset{1}{[0]}\cdots\overset{}{[0]}\overset{y}{[1]}\overset{}{[0]}\cdots\overset{n}{[0]}\Big) \enspace .
    \end{equation*}
Each $\mathsf{S}_i$ then multiplies $\vec{D}^{veh_{u_o}}$ and $\vec{D}^{u_o}$ to construct a vector of length $3$, i.e.,
    \begin{equation*}
        \vec{D}^{veh_{u_o}}\times\vec{D}^{u_o} = \Big(ID^{u_o}\; [ID^{veh_{u_o}}_y]\; [K^{veh_{u_o}}_y]\Big) \enspace .
    \end{equation*}
Based on the resultant vector, $\vec{D}^{veh_{u_o}}\times\vec{D}^{u_o}$, the secret key share of the vehicle $[K^{veh_{u_o}}_y]$ is retrieved. 

To preserve the confidentiality of $[M^{u_c}]$, each $\mathsf{S}_i$ encrypts it with the $[K^{veh_{u_o}}_y]$ using the symmetric key encryption, $\mathsf{E}()$, function. The generated \gls{AT} requires a second layer of encryption making $[AT^{veh_{u_o}}]$ and the $[ID^{veh_{u_o}}]$ available only to $u_c$. Specifically, the \gls{VSSP} servers, $\mathsf{S}_i$, collaboratively encrypt $[M^{u_c}]$ using the retrieved $[K^{veh_{u_o}}]$ to generate an \gls{AT} for the vehicle in shared form, i.e., $[AT^{veh_{u_o}}]$. Then, each $S_i$ collaboratively perform a second layer of encryption, using $[AT^{veh_{u_o}}]$ and $[ID^{veh_{u_o}}]$ with $[K^{u_c}_{enc}]$ to generate and retrieve $C^{u_c}$ using $\mathsf{open}([C^{u_c}])$.


In addition, each $\mathsf{S}_i$ generates an authentication tag, $[AuthTag^{BD^{u_o, u_c}}]$, that can be later used to retrieve the associated $AT^{veh_{u_o}}$ from the \gls{PL} by $u_c$. Using $\mathsf{mac}()$ with $[\vec{K}^{u_c}_{tag}]$ and $[BD^{u_o, u_c}]$ as inputs, each $\mathsf{S}_i$ creates an authentication tag $[AuthTag^{BD^{u_o, u_c}}]$.~\footnote{Recall that $\vec{K}^{u_c}_{tag}=(K^{u_c}_{tag_{mac}},K^{u_c}_{tag_{enc}})$.} Prior to posting on the \gls{PL}, we use $\mathsf{open}([AuthTag^{BD^{u_o, u_c}}])$, reconstructing the shares and obtain $AuthTag^{BD^{u_o, u_c}}$. Note that for the efficient \gls{MPC}, we perform Enc-then-Hash-then-MAC. The reason is that, following~\cite{rotaru2017modes} encryption, i.e., $\mathsf{E}()$, can be done in parallel and separately (thus efficient); the hash does not need to be done in \gls{MPC} and the \gls{MPC} parties, $\mathsf{S}_i$, can apply the hash function locally (see Sec.~\ref{sec:protocol_evaluation}). Essentially, we trade ``parallel \gls{MPC} encryption'' for ``having to evaluate a hash function on large input in \gls{MPC}''.~\footnote{In our implementation, we use CBC-MAC-AES and HtMAC-MiMC as we describe in Sec.~\ref{sec:protocol_evaluation}.}

Finally, each $\mathsf{S}_i$ sends an access-token-publication request, i.e., \textit{AT\_PUB\_REQ}, to \gls{PL} along with $C^{u_c}$ and $AuthTag^{BD^{u_o, u_c}}$. 


\begin{figure*}[hbt!]
    \centering
        \resizebox{0.95\textwidth}{!}{%
        \includegraphics{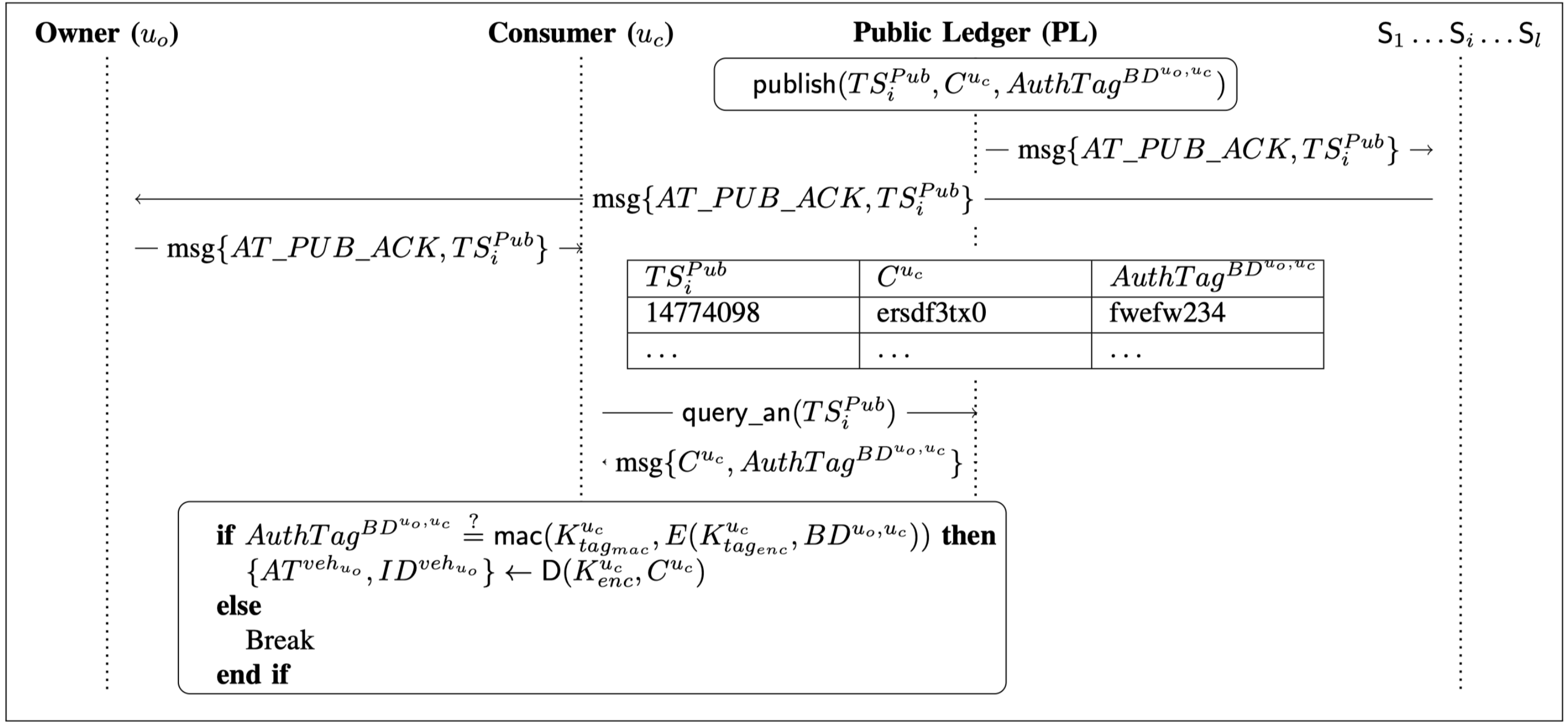}}
        \caption{Step 3: Access token distribution and verification.} 
    \label{fig:step3}
\end{figure*}

\paragraph*{Step 3 -- \acrfull{AT} distribution and verification}\label{step3}
The encrypted \gls{AT} is published at the \gls{PL}. The \gls{AT} then is retrieved by $u_c$ to access the vehicle. 

In detail, as depicted in Fig.~\ref{fig:step3}, after receiving the \textit{AT\_PUB\_REQ}, \gls{PL} publishes $C^{u_c}$, $AuthTag^{BD^{u_o, u_c}}$ and the publication time-stamp, i.e., $TS^{Pub}$.

The consumer, $u_c$, monitors \gls{PL} for concurrent and announced time-stamps, $TS^{Pub}$, to identify the corresponding $C^{u_c}$ using $AuthTag^{BD^{u_o, u_c}}$. Upon identification, $C^{u_c}$ queries and anonymously retrieves $C^{u_c}$ from \gls{PL} using $\mathsf{query\_an()}$, such that \gls{PL} cannot identify $u_c$. Then, $u_c$ decrypts $C^{u_c}$ using $K^{u_c}_{enc}$ to obtain the \gls{AT} and the vehicle identity, $\{AT^{veh_{u_o}}, ID^{veh_{u_o}}\}$.  Note that, in a parallel manner and for synchronization purposes, \gls{PL} forwards an acknowledgment of the publication, \textit{AT\_PUB\_ACK}, along with $TS^{Pub}_i$ to at least one $\mathsf{S}_i$ which then it forwards $TS^{Pub}_i$ to $u_c$ via $u_o$. Upon receipt of \textit{AT\_PUB\_ACK}, $u_c$ uses $TS^{Pub}_i$ to query $\mathsf{PL}$. In the same manner, it uses $\mathsf{query\_an()}$ to anonymously retrieve $C^{u_c}$ and $AuthTag^{BD^{u_o, u_c}}$. 

Then, $u_c$ verifies locally the authentication tag $C^{B}$ using the $\vec{K}^{u_c}_{tag}$ and $BD^{u_o, u_c}$ as inputs to the $\mathsf{mac}()$ function. A successful verification assures $u_c$ the validity of \gls{AT}, that it contains the agreed \gls{BD} during \textit{vehicle booking} prerequisite step. Next, $u_c$ using $K^{u_c}_{tag_{enc}}$ decrypts $C^{u_c}$ to retrieve,  $\{AT^{veh_{u_o}}, ID^{veh_{u_o}}\}$, the access token and the identifier of the vehicle respectively.


\begin{figure*}[hbt!]
    \centering
        \resizebox{0.95\textwidth}{!}{%
        \includegraphics{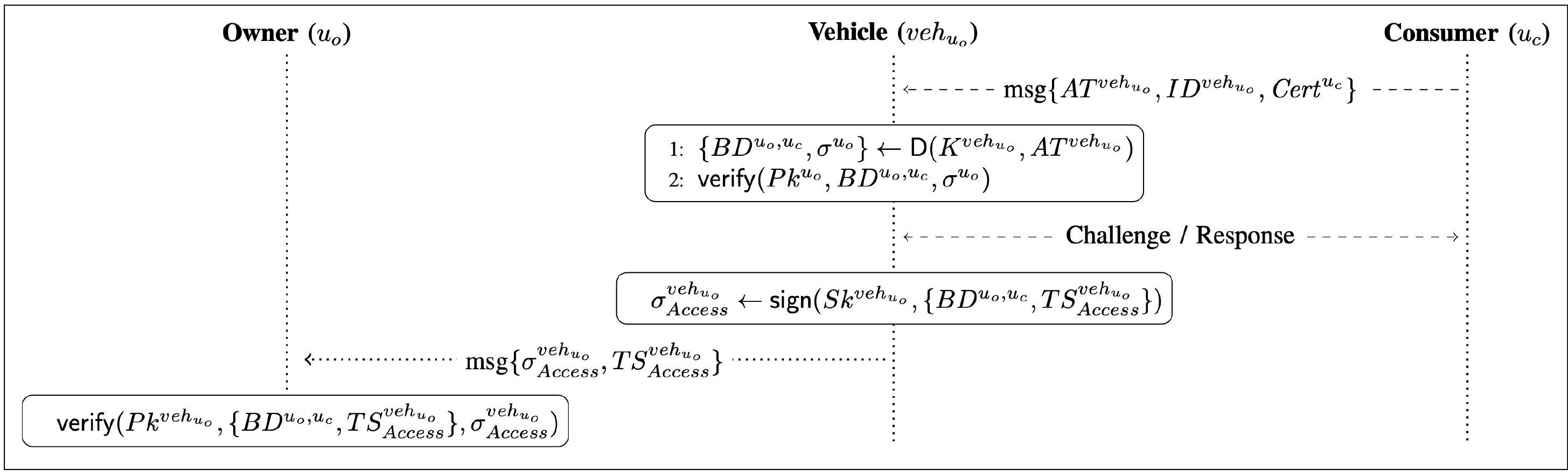}}
        \caption{Step 4: vehicle access. Dashed lines represent close range wireless communication.} 
    \label{fig:step4}
\end{figure*}

\paragraph*{Step 4 -- Car Access}\label{step4}
The consumer uses the $AT^{veh_{u_o}}$, $ID^{veh_{u_o}}$, and $\mathit{Cert}^{u_c}$, to obtain access to the vehicle, using any challenge-response protocol based on public key implementations~\cite{DBLP:conf/codaspy/DmitrienkoP17,DBLP:journals/dcc/DiffieOW92} (see Fig.~\ref{fig:step4}). 

In detail, $u_c$ sends directly to the vehicle $\{AT^{veh_{u_o}}, ID^{veh_{u_o}}, \mathit{Cert}^{u_c}\}$, using a secure and authenticated short-range communication channel such as NFC. It can use any challenge-response protocol for the connection establishment based on public/private key~\cite{DBLP:conf/codaspy/DmitrienkoP17,DBLP:journals/dcc/DiffieOW92}. Upon receipt, the \gls{OBU} of vehicle decrypts $AT^{veh_{u_o}}$ using $K^{veh_{u_o}}$ to obtain $M^{u_c} = \{BD^{u_o, u_c}, \sigma^{u_o}\}$. 

The \gls{OBU} then performs the following verification. First, it checks the signature $\sigma^{u_o}$ to verify that the booking details, $BD^{u_o, u_c}$, were not altered and were indeed approved by the vehicle owner. Then, it verifies the identity of $u_c$, using the received $\mathit{Cert}^{u_c}$ (along with the $\mathsf{hash}(\mathit{Cert}^{u_c})$ in $BD^{u_o, u_c}$). Finally, it verifies that the access attempt satisfies the conditions specified in $BD^{u_o, u_c}$. If successful, the \gls{OBU} grants $u_c$ access to $veh_{u_o}$. It signs $\{ BD^{u_o, u_c}, TS^{veh_{u_o}}_{Access} \}$, where $TS^{veh_{u_o}}_{Access}$ is the time-stamp of the instant at which access was granted to $veh_{u_o}$. Finally, it forwards the msg$\{\sigma^{veh_{u_o}}_{Access},TS^{veh_{u_o}}_{Access}\}$ to $u_o$. Otherwise, if any verification fails, the \gls{OBU} terminates the vehicle access process, denying access to the vehicle.


\section{Functional, Security and Privacy Requirements Analysis}\label{sec:extended_analysis}
\newcommand{\getsR}{\xleftarrow{{\scriptscriptstyle\$}}}
\newcommand{\advsign}[1]{\mathrm{Adv}_{\mathrm{sign}}^{#1\text{-}\mathrm{euf}}}
\newcommand{\advprf}[1]{ \mathrm{Adv}_{\mathrm{prf}}^{#1\text{-}\mathrm{prf}}}
\newcommand{\advpke}[1]{ \mathrm{Adv}_{\mathrm{enc}}^{#1\text{-}\mathrm{pke}}}
\newcommand{\advske}[1]{ \mathrm{Adv}_{\mathrm{E}}^{#1\text{-}\mathrm{ske}}}

\newcommand{\advmac}[1]{ \mathrm{Adv}_{\mathrm{mac}}^{#1\text{-}\mathrm{mac}}}
\newcommand{\advcol}{    \mathrm{Adv}_{\mathsf{hash}}^{\mathrm{col}}}
\newcommand{\badA}{\mathcal{A}}
\newcommand{\badAsign}{\mathcal{A}_{\mathrm{sign}}}
\newcommand{\badAprf}{\mathcal{A}_{\mathrm{prf}}}
\newcommand{\badApke}{\mathcal{A}_{\mathrm{pke}}}
\newcommand{\badAske}{\mathcal{A}_{\mathrm{ske}}}
\newcommand{\badAmac}{\mathcal{A}_{\mathrm{mac}}}
\newcommand{\eve}{\mathsf{E}}

We argue that \system\ fulfills its functional requirements and prove that it is secure and privacy-enhancing satisfying the requirements of Section~\ref{sec:system_model}.

\subsection{Functional Requirements Realization}



\paragraph*{FR1 -- Offline vehicle access} While Steps 1-3 (Fig.~\ref{fig:step1} - Fig.~\ref{fig:step3}) require a network connection, Step 4 provides vehicle access using short-range wireless communication. The vehicle can offline decrypt and verify the \gls{AT} using its key, $K^{veh_{u_o}}$, and the public-key, $Pk^{u_o}$, of $u_o$, both stored locally. The signature of access confirmation, $\sigma^{veh_{u_o}}_{Access}$, can be sent over the Internet to $u_o$, or when $veh_{u_o}$ and $u_o$ are in close proximity.

\paragraph*{FR2 -- \acrfull{AT} update and revocation by the owner $u_o$} \system\ can update or revoke \gls{AT} as described in Steps 1-3, as a new booking request. After an agreement for an update action between $u_o$ and $u_c$, the necessary \gls{BD} values are updated to $\hat{BD}^{u_o, u_c}$. In case of revocation, upon agreement between $u_o$ and $u_c$, the parameters in $\hat{BD}^{u_o, u_c}$ are set to a predefined value specifying the revocation action. There might be occasions in which the \gls{AT} update or revocation needs to be enforced by $u_o$ while preventing $u_c$ from blocking such requests/operations. \system\ can execute requests initiated by $u_o$ alone, without the involvement of $u_c$. More specifically, the generation of session keys is performed by $u_o$, requesting an \gls{AT} from \gls{VSSP}, querying the \gls{PL}, and forwarding the token to vehicle, $veh_{u_o}$. The \gls{PD} of the owner forwards the updated \gls{AT} using a short-range (in close proximity) or an Internet connection (e.g., cellular data) if needed, for restricting access a dishonest consumer (e.g., fleeing with the vehicle).

\subsection{Security and Privacy}

\system\ is secure and privacy-enhancing, provided that its underlying cryptographic primitives are sufficiently secure. Informally, we demonstrate the following:
\begin{theorem}[Informal]
    Assume that communication between all entities at \gls{VSS} takes place over private channels - are secure and authenticated using, e.g., SSL-TLS~\cite{rfc8446}. If
    \begin{itemize}
    \item the \gls{MPC} is statistically secure~\cite{rotaru2017modes},
     \item the key derivation function $\mathsf{kdf}$ is multi-key secure~\cite{DBLP:journals/jacm/GoldreichGM86},
    \item the signature scheme $\mathsf{sign}$ is multi-key existentially unforgeable~\cite{DBLP:journals/siamcomp/GoldwasserMR88},
    \item the public-key encryption scheme $\mathsf{enc}$ is multi-key semantically secure~\cite{DBLP:conf/eurocrypt/BellareBM00},
    \item the symmetric key encryption scheme $\mathsf{E}$ is multi-key chosen-plaintext secure~\cite{DBLP:conf/focs/BellareDJR97},
    \item the \gls{MAC} function $\mathsf{mac}$ is multi-key existentially unforgeable~\cite{DBLP:journals/siamcomp/GoldwasserMR88}, and
    \item the hash function $\mathsf{hash}$ is collision resistant~\cite{DBLP:conf/fse/RogawayS04},
\end{itemize}
then, \system\ fulfills the security and privacy requirements of Sect.~\ref{sec:system_model}.
\end{theorem}
Details on the \textit{semantic security analysis} of \system\ is given below. More precisely, in Section~\ref{sec:analysis-model} we describe the security models of the cryptographic primitives. Then, the formal reasoning is given in Section~\ref{sec:analysis-proof}.

\subsection{Cryptographic Primitives}\label{sec:analysis-model}
Note that in \system, the cryptographic primitives are evaluated under different keys, and therefore we will need the security of the cryptographic primitives in the \emph{multi-key} setting. For example, $\mathsf{enc}$ is used for different keys, each for a different party in the \gls{VSSP}, $\mathsf{E}$ and $\mathsf{mac}$ are used for independent keys (i.e., session keys) for every fresh evaluation of the protocol; and $\mathsf{sign}$ is used by all owners, $u_o$, each with a different key. Bellare et al.~\cite{DBLP:conf/eurocrypt/BellareBM00} showed how public key encryption can be generalized to \emph{multi-key} security; the adaptation straightforwardly generalizes to the other security models. 

In the definitions below, for a function $f$, we define by $\mathrm{Func}(f)$ as the set of all functions with the exact same interface as $f_K$. We denote a random drawing by $\getsR$.

\begin{definition}\label{def:prf}
    Let $\mu\geq1$. Consider a key derivation function using a pseudorandom function $\mathrm{prf}=(\mathsf{kg},\mathsf{prf})$.
    We define the advantage of an adversary $\badA$ in breaking the $\mu$-multikey pseudorandom function security as
    \begin{align*}
    &\advprf{\mu}(\badA) = \\
    &\qquad\left| \Pr\left( K^1,\ldots,K^{\mu}\getsR\mathsf{kg} \;:\; \badA^{\mathsf{prf}(K^i,\cdot)}=1\right) -\right.\\
    &\qquad\qquad\qquad\left.\Pr\left( \$^1,\ldots,\$^{\mu}\getsR\mathrm{Func}(\mathsf{prf}) \;:\; \mathcal{A}^{\$^i}=1\right) \right|\,.
\end{align*}
    We define by $\advprf{\mu}(q,t)$ the maximum advantage, taken over all adversaries making at most $q$ queries and running in time at most $t$.
\end{definition}

\begin{definition}\label{def:sign}
    Let $\mu\geq1$. Consider a signature scheme $\mathrm{sign}=(\mathsf{kg},\mathsf{sign},\mathsf{verify})$.
    We define the advantage of adversary $\badA$ in breaking the $\mu$-multikey existential unforgeability as
    \begin{align*}
            &\advsign{\mu}(\badA) = \\
            &\qquad\Pr\left((Pk^1,Sk^1),
            \ldots,(Pk^{\mu},Sk^{\mu})\getsR\mathsf{kg} \;:\;\right.\\
            &\qquad\qquad\qquad\qquad\qquad\qquad\left.\badA^{\mathsf{sign}(Sk^i,\cdot)}(Pk^i) 
            \text{ forges} \right)\,,
    \end{align*}
    where ``forges'' means that $\badA$ outputs a tuple $(i,M,\sigma)$ such that $\mathsf{verify}(Pk^i,M,\sigma)=1$ and $M$ has never been queried to the $i$-th signing oracle. 
    We define by $\advsign{\mu}(q,t)$ the maximum advantage, taken over all adversaries making at most $q$ queries and running in time at most $t$.
\end{definition}

\begin{definition}\label{def:enc}
    Let $\mu\geq1$. Consider a public-key encryption scheme $\mathrm{enc}=(\mathsf{kg},\mathsf{enc},\mathsf{dec})$.
    We define the advantage of adversary $\badA$ in breaking the $\mu$-multikey semantic security as
    \begin{align*}
    &\advpke{\mu}(\badA) = \\
    &\left| \Pr\left( (Pk^1,Sk^1),\ldots,(Pk^{\mu},Sk^{\mu})\getsR\mathsf{kg} \;:\; \badA^{\mathcal{O}_0}(Pk^i)=1\right) -\right.
    \\
    &\left.\Pr\left( (Pk^1,Sk^1),\ldots,(Pk^{\mu},Sk^{\mu})\getsR\mathsf{kg} \;:\; \badA^{\mathcal{O}_1}(Pk^i)=1\right) \right|\,,
    \end{align*}
    where $\mathcal{O}_b$ for $b\in\{0,1\}$ gets as input a tuple $(i,m_0,m_1)$ with $i\in\{1,\ldots,\mu\}$ and $|m_0|=|m_1|$ and outputs $\mathsf{enc}_{Pk^i}(BD^{u_o, u_c})$. 
    We define by $\advpke{\mu}(t)$ the maximum advantage, taken over all adversaries running in time at most $t$.
\end{definition}
\begin{definition}\label{def:e}
    Let $\mu\geq1$. Consider a symmetric-key encryption scheme $\mathrm{E}=(\mathsf{kg},\mathsf{E},\mathsf{D})$.
    We define the advantage of adversary $\badA$ in breaking the $\mu$-multikey chosen-plaintext security as
    \begin{align*}
        & \advske{\mu}(\badA) = \\
        &\qquad\left| \Pr\left( K^1,\ldots,K^{\mu}\getsR\mathsf{kg} \;:\; \badA^{\mathsf{E}(K^i,\cdot)}=1\right) - \right.\\
        &\qquad\qquad\qquad\left.\Pr\left( \$^1,\ldots,\$^{\mu}\getsR\mathrm{Func}(\mathsf{E}) \;:\; \mathcal{A}^{\$^i}=1\right) \right|\,.
    \end{align*}
    We define by $\advske{\mu}(q,t)$ the maximum advantage, taken over all adversaries making at most $q$ queries and running in time at most $t$.
\end{definition}
\begin{definition}\label{def:mac}
    Let $\mu\geq1$. Consider a MAC function $\mathrm{mac}=(\mathsf{kg},\mathsf{mac})$.
    We define the advantage of adversary $\badA$ in breaking the $\mu$-multikey existential unforgeability as
    \begin{multline*}
    \advmac{\mu}(\badA) = \\\Pr\left( K^1,\ldots,K^{\mu}\getsR\mathsf{kg} \;:\; \badA^{\mathsf{mac}(K^i,\cdot)} \text{ forges} \right)\,,
    \end{multline*}
    where ``forges'' means that $\badA$ outputs a tuple $(i,M,\sigma)$ such that $\mathsf{mac}(K^i,M)=\sigma$ and $M$ has never been queried to the $i$-th MAC function. 
    We define by $\advmac{\mu}(q,t)$ the maximum advantage, taken over all adversaries making at most $q$ queries and running in time at most $t$.
\end{definition}
Finally, we consider the hash function $\mathsf{hash}$ to be collision-resistant. We denote the supremal probability of any adversary in finding a collision for $\mathsf{hash}$ in $t$ time by $\advcol(t)$. The definition is, acknowledgeably, debatable: for any hash function there exists an adversary that can output a collision in constant time (namely, one that has a collision hardwired in its code). We ignore this technicality for simplicity and refer to \cite{DBLP:conf/fse/RogawayS04,DBLP:journals/dcc/Stinson06,DBLP:conf/vietcrypt/Rogaway06} for further discussion.

\subsection{Analysis}\label{sec:analysis-proof}
We prove that \system\ satisfies the security and privacy requirements of Section~\ref{sec:system_model}, provided that its underlying cryptographic primitives are sufficiently secure.
\begin{theorem}\label{thmextended}
Suppose that communication takes place over private channels, the MPC is statistically secure, $\mathsf{hash}$ is a random oracle, and

\begin{equation*}
    \begin{split}
        \advsign{\mu_o+\mu_{\mathrm{veh_{u_o}}}}(2q,t) + \advprf{\mu_c}(2q,t) + \advpke{l}(t) + \\
        \advske{2q+\mu_\mathrm{veh_{u_o}}}(3q,t) + \advmac{q}(q,t) + \advcol(t) \ll1\,,
    \end{split}
\end{equation*}

where $\mu_o$ denotes the maximum number of $u_o$s, $\mu_c$ the maximum number of $u_c$s, $\mu_\mathrm{veh_{u_o}}$ the maximum number of vehicles, $l$ the number of servers in the \gls{VSSP}, $q$ the total times the system gets evaluated, and $t$ the maximum time of any adversary.

Then, \system\ fulfills the security and privacy requirements of Section~\ref{sec:system_model}.
\end{theorem}
\begin{proof}
Recall from Section~\ref{sec:system_model} that owners, $u_o$, and \gls{VM} are honest-but-curious, whereas consumers, $u_c$, and outsiders may be malicious and actively deviate from the protocol. Vehicles are trusted.

Via a hybrid argument, we replace the key derivation functions utilizing pseudorandom functions $\mathsf{prf}(K^{u_c},\cdot)$ by independent random functions $\$^{u_c}$. This step is performed at the cost of
\begin{align}
\advprf{\mu_c}(2q,t)\,,\label{eqn:cost1}
\end{align}
as in every of the $q$ evaluations of \system\ there are two evaluations of a function $\mathsf{prf}$, and at most $\mu_c$ instances of these functions. As we assume that the \gls{MPC} is statistically secure, we can replace the \gls{VSSP} by a single, trusted \acrfull{SP} (with $l$ interfaces) -- it perfectly evaluates the protocol, and it does not reveal/leak any information. Assuming that the public-key encryption reveals nothing, which can be done at the cost of
\begin{align}
\advpke{l}(t)\,,\label{eqn:cost2}
\end{align}
we can, for simplicity, replace it with a perfectly secure public-key encryption $\rho^{VSSP}$ at the \gls{VSSP} directly. Thus, an encryption does not reveal its origin and content, and only \gls{VSSP} can straightforwardly decrypt, therewith eliminating the fact that \gls{VSSP} has $l$ interfaces and has to perform multiparty computation. Now, as the pseudorandom functions are replaced by random functions, the keys to the symmetric encryption scheme, $\mathsf{E}$, are all independently and uniformly distributed, and as the public-key encryption scheme is secure, these keys never leak. Therefore, we can replace the symmetric encryption functionality by perfectly random invertible functions, $\pi^{veh_{u_o}}$ for the vehicles, unique $\pi^{u_c}_{enc}$ for every new encryption with the $u_c$ session keys, and $\pi^{u_c}_{tag_{enc}}$ for every new encryption in the tag computation with $u_c$ session keys, at the cost of
\begin{align}
\advske{2q+\mu_\mathrm{veh_{u_o}}}(3q,t)\,,\label{eqn:cost3}
\end{align}
as there are $2q+\mu_\mathrm{veh_{u_o}}$ different instances involved and at most $3q$ evaluations are made in total. This means that instead of randomly drawing $K^{u_c}_{enc} \leftarrow \$^{u_c}$, we now randomly draw $\pi^{u_c}_{enc}\getsR \mathrm{Func}(\mathsf{E})$. Likewise, for $K^{u_c}_{tag_{enc}}\leftarrow \$^{u_c}$ we now randomly draw $\pi^{u_c}_{tag_{enc}}\getsR \mathrm{Func}(\mathsf{E})$.

We are left with a simplified version of \system. The \gls{VSSP} is replaced by a single trusted authority. The pseudorandom functions are replaced by independent random drawings - $u_c$ uses $\$^{u_c}$ which generates fresh outputs for every call. The public-key encryptions are replaced with a perfectly secure public-key encryption function $\rho^{VSSP}$. Finally, the symmetric-key encryptions are replaced by perfectly random invertible functions $\pi^{veh_{u_o}}$, $\pi^{u_c}_{enc}$, and $\pi^{u_c}_{tag_{enc}}$. The simplified system is illustrated in Figure~\ref{fig:protsimplified}. Here, the derivation of the vehicle key (or, formally, the random function corresponding to the encryption) from the database is abbreviated to $\pi^{veh_{u_o}} \leftarrow \mathsf{query}(ID^{u_o}, DB^{S_i})$ for conciseness.

\bigskip\noindent We will now treat the security and privacy requirements, and discuss how these are achieved from the cryptographic primitives, separately. We recall that the consumer $u_c$ and owner $u_o$ have agreed upon the \gls{BD} prior to the evaluation of \system, hence they know each other by design.

\paragraph*{SR1 --  Confidentiality of \acrfull{BD}, $BD^{u_o, u_c}$} 

In one evaluation of the protocol, $u_c$, $u_o$, the \emph{trusted} \gls{VSSP}, and the shared vehicle, $veh_{u_o}$ learn the \gls{BD} by default or by design. \gls{BD} can only become public through the values $AuthTag^{BD^{u_o, u_c}}$ and $C^{u_c}$ satisfying
\begin{equation}\label{eqn:AuthT}
    \begin{split}
        AuthTag^{BD^{u_o, u_c}} & = \mathsf{mac}(K^{u_c}_{tag_{mac}},\mathsf{E}(K^{u_c}_{tag_{enc}},BD^{u_o, u_c}) \\
    & = \mathsf{mac}(K^{u_c}_{tag_{mac}},\pi^{u_c}_{tag_{enc}}(BD^{u_o, u_c}))\,,    
    \end{split}
\end{equation}

\begin{equation}\label{eqn:Cuc}
    \begin{split}
        C^{u_c} & = \mathsf{E}(K_{enc}^{u_c},\{\mathsf{E}(K^{veh_{u_o}}_y,\{BD^{u_o, u_c},\sigma^{u_o}\}),ID^{veh_{u_o}}\}) \\
        & = \pi^{u_c}_{enc}(\{\pi^{veh_{u_o}}(\{BD^{u_o, u_c},\sigma^{u_o}\}),ID^{veh_{u_o}}\})\,.   
    \end{split}
\end{equation}

Eqn:~\ref{eqn:AuthT},~\ref{eqn:Cuc} reveal nothing about $BD^{u_o, u_c}$ to a malicious outsider, thanks to the security of $\mathsf{mac}$, $\mathsf{E}$, and the independent uniform drawing of the keys $K^{u_c}_{enc}$ and $\vec{K}^{u_c}_{tag}=(K^{u_c}_{tag_{enc}}, K^{u_c}_{tag_{mac}})$; $\pi^{u_c}_{enc}$ and $\pi^{u_c}_{tag_{enc}}$ randomly generated for every evaluation.

The nested encryption $\mathsf{E}$, i.e., $\pi^{u_c}_{enc}\circ \pi^{veh_{u_o}}$ (Eqn:~\ref{eqn:Cuc}), does not influence the analysis due to the mutual independence of the two functions, i.e. the mutual independence of the keys $K_{enc}^{u_c}$ and $K^{veh_{u_o}}_y$.

\paragraph*{SR2 -- Entity and data authenticity of $BD^{u_o, u_c}$ from the owner $u_o$}
An owner who initiates the \gls{AT} generation and distribution, first signs the \gls{BD} using its private key before sending those to the \gls{VSSP} in shares. Therefore, once the vehicle receives the token and obtains the booking details, it can verify the signature of $u_o$ on $BD^{u_o, u_c}$. In other words, the vehicle can verify the source of $BD^{u_o, u_c}$, $u_o$, and its integrity. Suppose, to the contrary, that a malicious consumer can get access to a vehicle of an $u_o$. This particularly means that it created a tuple $(BD^{u_o, u_c},\sigma^{u_o})$ such that $\mathsf{verify}(Pk^{u_o},BD^{u_o, u_c},\sigma^{u_o})$ holds. If $\sigma^{u_o}$ is new, this means that $u_c$ forges a signature for the secret signing key $\mathit{Sk}^{u_o}$. Denote the event of this happening by
\begin{equation}\label{eqn:event1}
    \begin{split}
        &\eve_1\;:\;\badA \text{ forges } \mathsf{sign}(Sk^{u_o},\cdot) \text{ for some }Sk^{u_o}\,. 
    \end{split}
\end{equation}

On the other hand, if $(BD^{u_o, u_c},\sigma^{u_o})$ is old but the evaluation is fresh, this means a collision $\mathsf{hash}(\mathit{Cert}^{u_c})=\mathsf{hash}(\mathit{Cert}^{u_c\prime})$. Denote the event of this happening by
\begin{align}
&\eve_2\;:\;\badA \text{ finds a collision for } \mathsf{hash}\,. \label{eqn:event2}
\end{align}
We thus obtain that a violation of SR2 implies $\eve_1\vee\eve_2$.


\paragraph*{SR3 -- Confidentiality of $AT^{veh_{u_o}}$} The \gls{AT} is generated by the \gls{VSSP} obliviously - as the \gls{VSSP} is trusted. The \gls{AT} is only revealed to the public in encrypted form, through $C^{u_c}$ of (\ref{eqn:Cuc}). Due to the uniform drawing of $\pi^{u_c}_{enc}$ (and the security of $\rho^{VSSP}$ used to transmit this function), only the legitimate user (i.e., $u_c$) can decrypt and learn the \gls{AT}. It shares it with the vehicle over a secure and private channel.

\paragraph*{SR4 --  Confidentiality of vehicle key, $K^{veh_{u_o}}$}
By virtue of our hybrid argument on the use of the symmetric-key encryption scheme, $\mathsf{E}_{K^{veh_{u_o}}}$ got replaced with $\pi^{veh_{u_o}}$, which itself is a keyless random encryption scheme. As the key is now absent, it cannot leak.

Moreover, only the \acrfull{VM} and the vehicle itself hold copies of the vehicle key. The \gls{VM}, as a trusted \gls{SP}, holds all the secret keys of vehicles. As vehicle owners register their vehicles, the \gls{VM} forwards the list of $ID^{veh_{u_o}}$ to \gls{VSSP}. Each \gls{VSSP} server receives  $K^{veh_{u_o}}$ in secret shared form; is indistinguishable from randomness. Hence, these servers learn nothing about the vehicle secret key by virtue of the statistical security of the \gls{MPC}. 

In a nutshell, to retrieve the $y$th key from $DB^{\mathsf{S}_i}$, i.e., $[K^{veh_{u_o}}_y]$, each $\mathsf{S}_i$ performs an equality check over \gls{MPC}. The comparison outcomes 0 for mismatch and 1 for identifying the vehicle at position $y$, i.e.,
 \begin{equation*}
        \vec{D}^{veh_{u_o}}=
        \Big(\overset{1}{[0]}\cdots\overset{}{[0]}\overset{y}{[1]}\overset{}{[0]}\cdots\overset{n}{[0]}\Big) \enspace .
    \end{equation*}
from which the share of the vehicle's secret key, $[K^{veh_{u_o}}]$, can be retrieved. Due to the properties of threshold secret sharing, the secret vehicle keys stay secret to each $\mathsf{S}_i$. Thus, among all \gls{VSS} entities, only the \gls{VM} and the vehicle hold the vehicle key.

\paragraph*{SR5 -- Backward and forward secrecy of $AT^{veh_{u_o}}$}
The \gls{AT} is published on the \acrfull{PL} as $C^{u_c}$ of (\ref{eqn:Cuc}), encrypted using $\pi^{u_c}_{enc}$ (i.e., symmetric key $K_{enc}^{u_c}$). Every honest $u_c$ generates a uniformly randomly drawn function $\pi^{u_c}_{enc}$ (a fresh key $K_{enc}^{u_c}$) for every new evaluation. It uses a key derivation function $\mathsf{kdf}$ utilizing a \gls{PRF} for each key generation and every new evaluation of the protocol, and that is secure. This implies that all session keys are drawn independently and uniformly at random. In addition, the symmetric encryption scheme $\mathsf{E}$ is multi-key secure. Thus, all encryptions $C^{u_c}$ are independent and reveal nothing of each other. Note that nothing can be said about \glspl{AT} for malicious users who may deviate from the protocol and reuse one-time keys.

\paragraph*{SR6 -- Non-repudiation of origin of $AT^{veh_{u_o}}$}
The vehicle, who is a trusted entity, verifies the origin through verification of the signature, i.e., $\mathsf{verify}(Pk^{u_o},BD^{u_o, u_c},\sigma^{u_o})$. The consumer $u_c$ verifies the origin through the verification of the \gls{MAC} function, i.e.,
\begin{equation*}
    AuthTag^{BD^{u_o, u_c}} \stackrel{?}{=} \mathsf{mac}(K^{u_c}_{tag_{mac}},\pi^{u_c}_{tag_{enc}}(BD^{u_o, u_c}))\,.
\end{equation*}

Note that $u_c$ does not effectively verify $AT^{veh_{u_o}}$ but rather $AuthTag^{BD^{u_o, u_c}}$, which suffices under the assumption that the \gls{MPC} servers evaluate their protocol correctly. In either case, security fails only if the asymmetric signature scheme or the \gls{MAC} function are forgeable. The former is already captured by event $\eve_1$ in (\ref{eqn:event1}). For the latter, denote the event this happens by
\begin{align}
&\eve_3\;:\;\badA \text{ forges }
\mathsf{mac}(K^{u_c}_{tag_{mac}},\cdot) \text{ for some }K^{u_c}_{tag_{mac}}\,. \label{eqn:event3}
\end{align}
We thus obtain that a violation of SR6 implies $\eve_1\vee\eve_3$.

\paragraph*{SR7 -- Non-repudiation of $AT^{veh_{u_o}}$ receipt by $veh_{u_o}$ at $u_o$}
The owner $u_o$ can verify the correct delivery of $AT^{veh_{u_o}}$  with the successful verification and message sent by the vehicle to the owner, $\mathsf{verify}(Pk^{veh_{u_o}},\{BD^{u_o, u_c},TS^{veh_{u_o}}_{Access}\}, \sigma^{veh_{u_o}}_{Access})$ at the end of the protocol. Security breaks only if the signature scheme is forgeable. Denote the event of this happening by
\begin{align}
&\eve_4\;:\;\badA \text{ forges } \mathsf{sign}(Sk^{veh_{u_o}},\cdot) \text{ for some }Sk^{veh_{u_o}}\,. \label{eqn:event4}
\end{align}
We thus obtain that a violation of SR7 implies $\eve_4$.

\paragraph*{SR8 -- Accountability of users (i.e., owner $u_o$ and consumer $u_c$)} 
In case of wrongdoing or a dispute, a specific transaction may need to be retrieved and its information reconstructed (and only this information). Reconstruction of information is possible under the condition that \gls{VSSP} servers collude to reveal the shares of a transaction. However, these servers in our setting have competing interests. They would not collaborate and collude to reveal the shares of a transaction unless on wrongdoings were there is a request from legitimate entities such as law authorities'. In our scenario, the private inputs, i.e., information of transactions, can be reconstructed by a majority coalition due to threshold secret sharing properties~\cite{DBLP:conf/ccs/ArakiFLNO16,DBLP:journals/cacm/Shamir79}. This is, if the \gls{VSSP} consists of three servers, it suffices two of the server-shares required to reconstruct the secret.

\paragraph*{PR1 -- Unlinkability of (any two) requests of any consumer, $u_c$, and the vehicle, $veh_{u_o}$(s)}
Consumer and vehicle-identifiable data are included only in $BD^{u_o, u_c}$. It contains the certificate of $u_c$, $\mathit{Cert}^{u_c}$, and the identities of $u_c$, $ID^{u_c}$ and vehicle, $ID^{veh_{u_o}}$. Recall that $BD^{u_o, u_c}$ data are agreed between $u_c$ and $u_o$ before \system\ commences, so $u_o$ learns the identity of $u_c$ by default (prerequisite \textit{Step~B: vehicle booking} in Sec.~\ref{sec:system}). Beyond that, $u_c$ communicates only with the vehicle, $veh_{u_o}$, to forward the $AT^{veh_{u_o}}$ and perform access control. The consumer, $u_c$, queries the $AT^{veh_{u_o}}$ using an anonymous communication channel such ag Tor~\cite{torproject}. The \gls{BD} data are exchanged with the \gls{VSSP} encrypted and do not leak information by virtue of their confidentiality (security requirement SR1).

\paragraph*{PR2 -- Anonymity of any consumer, $u_c$, and vehicle, $veh_{u_o}$} The reasoning is identical to that of PR1.

\paragraph*{PR3 -- Indistinguishability of $AT^{veh_{u_o}}$ operations} \system\ utilizes the same steps and type of messages to \gls{VSSP} and \gls{PL} for access token generation, update, or revocation operation. Hence, system entities and outsiders can not distinguish which type of operation has been requested.


\paragraph{Conclusion} \system\ operates securely as long as the costs of (\ref{eqn:cost1}-\ref{eqn:cost3}), together with the probability that one of the events (\ref{eqn:event1}-\ref{eqn:event4}) occurs, are sufficiently small:
\begin{align*}
& \advprf{\mu_c}(2q,t) + \advpke{l}(t) \:+\\
&\qquad\qquad\advske{2q+\mu_\mathrm{veh_{u_o}}}(3q,t) \:+\\
&\qquad\qquad\qquad\qquad\Pr\left(\eve_{1}\vee\eve_{2}\vee\eve_{3}\vee\eve{4}\right) \ll 1\,.
\end{align*}
By design, the probability that the event $\eve_1\vee\eve_4$ occurs is upper bounded by $\advsign{\mu_o+\mu_{\mathrm{veh_{u_o}}}}(2q,t)$; the probability that event $\eve_3$ occurs is upper bounded by $\advmac{q}(q,t)$, and the probability that $\eve_2$ occurs is upper bounded by $\advcol(t)$. We thus obtain:
\begin{align*}
    &\Pr\left(\eve_{1}\vee\eve_{2}\vee\eve_{3}\vee\eve{4}\right)\\ 
    \leq\:& \advsign{\mu_o+\mu_{\mathrm{veh_{u_o}}}}(2q,t) + \advmac{q}(q,t) + \advcol(t)\,,
\end{align*}
which completes the proof.
\end{proof}

\section{Performance Evaluation and Analysis}\label{sec:protocol_evaluation}

\begin{table*}[ht!]
  \caption {\system\ performance: efficiency and scalability improvements, i.e., \acrfull{AT} generation, per number of vehicles utilizing: CBC-MAC-AES and HtMAC-MiMC. Throughput is evaluated for all servers and communication cost per server.}
  \label{tab:dabit}
\newcolumntype{C}{>{\centering\arraybackslash}X}
\begin{center}
  \begin{tabularx}{\textwidth}{lCCCCC}
    \toprule \toprule
   Type of Vehicle Owners & Protocol & Number of Vehicles per Owner & Communication Rounds & Communication Data (kB) & Throughput (ops/s) \\
	\midrule
     & CBC-MAC-AES & 1   & 568 & 64 & 33 \\ 
    & HtMAC-MiMC        & 1 & 167 & 108 & 546 \\
    \cmidrule{2-6}
    \multirow{3}{*}{Individuals} & CBC-MAC-AES & 2 &  568 & 64 & 32 \\ 
    & HtMAC-MiMC        & 2  & 167 & 108 & 546 \\
    \cmidrule{2-6}
    & CBC-MAC-AES & 4 &  568 &  107.7 & 32 \\ 
    & HtMAC-MiMC        & 4  & 167 & 117 & 544 \\
	\midrule
     & CBC-MAC-AES & 256 & 568 &  76 & 32 \\ 
    & HtMAC-MiMC  & 256  & 167 & 150 & 260 \\
    \cmidrule{2-6}
    \multirow{3}{*}{Vehicle-rental company branches} & CBC-MAC-AES & 512 & 568 &  88 & 32 \\ 
    & HtMAC-MiMC  & 512 & 167 & 194 & 151 \\

    \cmidrule{2-6}
    & CBC-MAC-AES & 1024 & 568 & 112 & 32 \\ 
    & HtMAC-MiMC  & 1024 & 167 & 280 & 84 \\
 
    \bottomrule \bottomrule
\end{tabularx}
\end{center}
\end{table*}

We argue that \system\ fulfills its performance requirement for efficiency and scalability in supporting a large volume of vehicles per user as in real-world deployment (see. Sect.~\ref{sec:system_model}).

 \begin{figure}[t]
        \centering
            \resizebox{\columnwidth}{!}{%
            \includegraphics{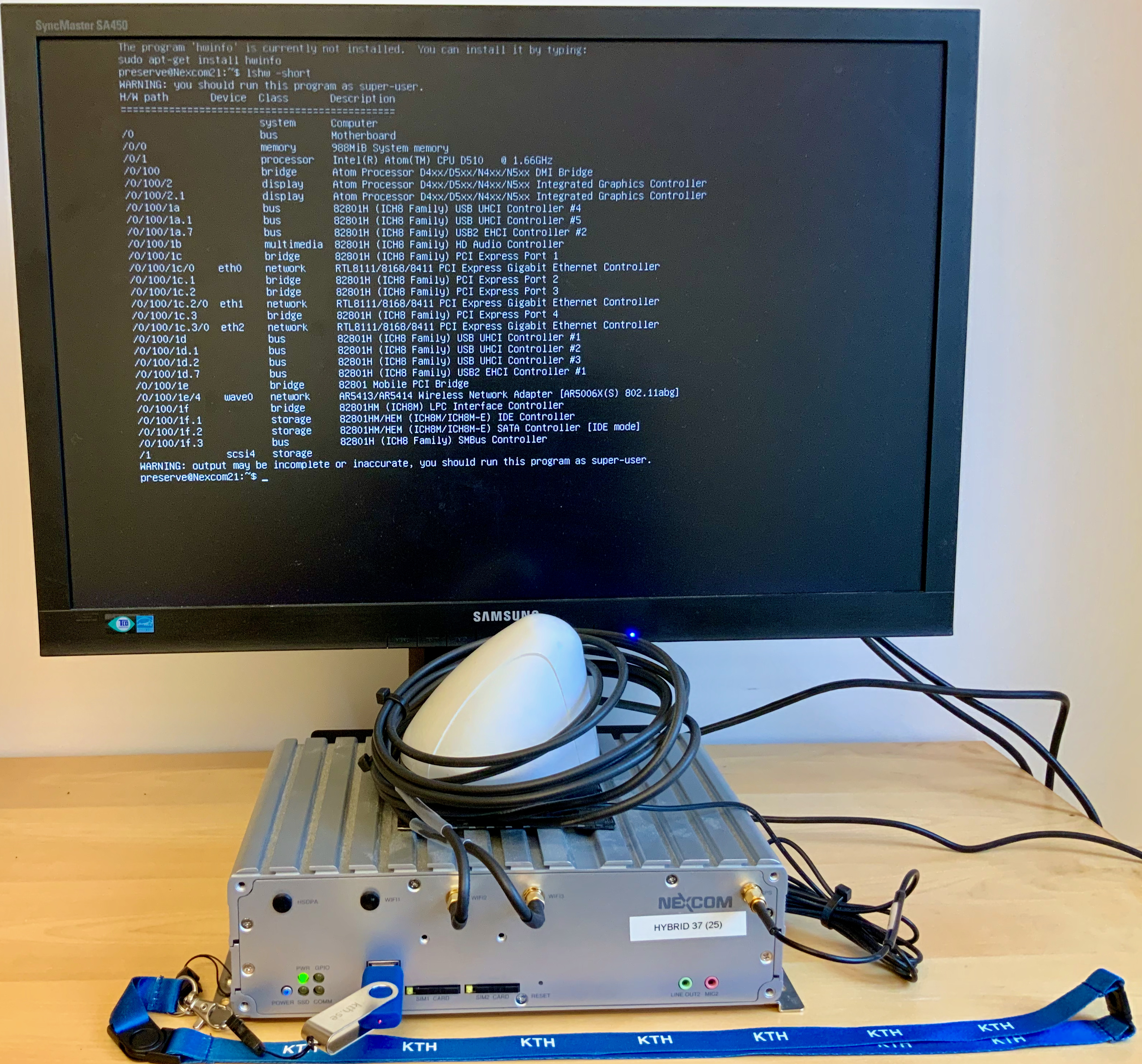}}
            \caption{Nexcom vehicular \acrfull{OBU} box~\cite{nexcom:vtc6201-ft}.} 
        \label{nexcom}
\end{figure}
    
\subsection{Benchmark and environment settings}
In \system\ we take a different approach to \exprotocol~\cite{DBLP:conf/esorics/SymeonidisAMMDP17} as we implement our protocols
in a fully-fledged open-sourced \gls{MPC} framework, i.e., MP-SPDZ \cite{DBLP:conf/ccs/Keller20} (in Step~2 -- see Fig.~\ref{fig:step4}).
The framework supports more than $30$ \gls{MPC} protocols
carefully implemented in \CC. In addition, a Python front-end
compiler allows the expression of circuits in a relatively simple way.
For MP-SPDZ, the compiler reduces the high-level program description
to bytecode or set of instructions for which the parties then 
run an optimized virtual machine written in \CC\ to execute the protocols.
In our case, two versions of \system\ were benchmarked: one with 
CBC-MAC tailored for binary circuits, while the other one uses HtMAC,
which is tailored for arithmetic circuits. We deployed the cryptographic operations in Step~1, Step~3 and Step~4 with OpenSSl~\cite{openssl} and python script for the secret sharing implementation~\cite{DBLP:conf/ccs/ArakiFLNO16}.

For our benchmark we use the following settings: $\length{M^{u_c}} = \length{AT^{veh_{u_o}}} = 10 \cdot
128$-bits, whereas $ID^{veh_{u_o}} \leq 2^{32}$, which thus fits into one $128$
bit-string, and $\length{BD^{u_o, u_c}} = 6 \cdot 128$-bits (including padding). Specifically, we consider $BD^{u_o, u_c}$ with the following message configuration-size: the vehicle identifier $ID^{veh_{u_o}}$ of $32$-bits, the location of the vehicle $L^{veh_{u_o}}$ of $64$-bits, the hashed certificate value of $u_c$ $\mathsf{hash}(Cert^{u_c})$ of $512$-bits, the \gls{BD} identifier $ID^{BD}$ of $32$-bits, the conditions and access rights accessing a vehicle by $u_c$, $CD^{u_c}$ of $96$-bits, and $AC^{u_c}$ of $8$-bits respectively. The $BD^{u_o, u_c}$ signature, $\sigma^{u_o}$, that $u_o$ will provide is of $2048$-bits using RSA-PKCS $\#1$ v2.0~\cite{rfc_2437}.

\paragraph{Environment Settings} We benchmarked our protocols using three distinct computers connected on a LAN network equipped with Intel i$7$-$7700$ CPU
with $3,60$~GHz and $32$~GB of RAM. For intra-\gls{VSSP} communication, we consider $10$~Gb/s network switch and $0.5$~ms \gls{RTT}.~\footnote{The implementation can be obtained from: https://github.com/rdragos/MP-SPDZ/tree/hermes} The vehicular \gls{OBU} \textit{Nexcom VTC 6201-FT} box (see Fig.~\ref{nexcom}) is used to benchmark Step~4. It is equipped with an Intel Atom-D$510$ CPU with $1,66$~GHz and $1$~GB of RAM~\cite{nexcom:vtc6201-ft} from the PRESERVE project~\cite{preserve}.~\footnote{The prototype \gls{OBU} is a hardware module without any additional cryptographic hardware accelerator.}


\subsection{Theoretical Complexity} Measuring the complexity of an \gls{MPC}
protocol usually boils down to counting the number of non-linear operations
in the circuit and the circuit depth. We consider the case where
a protocol is split into two phases. An input-independent (preprocessing)
 phase, where the goal is to produce correlated
randomness. Additionally, an input-dependent (online) phase where parties in \gls{VSSP} provide their
inputs and start exchanging data using the correlated randomness produced
beforehand. One secret multiplication (or an AND gate for the
$\mathbb{F}_2$ case) requires one random Beaver triple (correlated
randomness)~\cite{DBLP:conf/crypto/Beaver91a} from the 
preprocessing phase and two $\mathsf{open()}$ operations in the online phase.

Note that, in our case, the two versions of \system\ are benchmarked using the following two executables: \textit{replicated-bin-party.x} ($\mathbb{F}_2$ case, CBC-MAC)
and \textit{replicated-field-party.x} ($\mathbb{F}_p$ case, HtMAC).
The first executable is the implementation of Araki et al.'s binary-based
protocol~\cite{DBLP:conf/ccs/ArakiFLNO16}, while the latter is for the field case.
Next, we analyze the complexity of these two separately and motivate the two choices.


\paragraph{\textit{CBC-MAC-AES -- Case for Binary Circuits}}
This solution is implemented to have a baseline comparison with \exprotocol~\cite{DBLP:conf/esorics/SymeonidisAMMDP17} using
MP-SPDZ~\cite{DBLP:conf/ccs/Keller20}. The equality check is implemented using a binary tree of AND operations
with a $\log{n}$ depth where $n$ is the number of vehicles (see Fig.~\ref{fig:step2}).
Obtaining the corresponding vehicle key $[K^{veh_{u_o}}]$ assuming there are $n$ vehicles per user
has a cost of $159 \cdot n$ Beaver triples assuming $32$ bit length vehicle IDs. 

When evaluating the operations depicted in~Fig.~\ref{fig:step2} in \gls{MPC},
the most expensive part is computing $[AT^{veh_{u_o}}]$ since that requires
encrypting $10 \cdot 128$ bits, calling AES $10$ times, which has a cost of
$6400 \cdot 10$ AND gates. In the next step, (in line~$8$, Fig.~\ref{fig:step2}), AES is called $11$ times, while the operation computing
CBC-MAC-AES (in line~$10$, Fig.~\ref{fig:step2}) takes only $6$ AES calls. Given
the above breakdown, the theoretical cost for generating an \acrfull{AT} has
a cost of $159 \cdot n + 6400 \cdot 28$ AND gates.


\paragraph{\textit{HtMAC-MiMC -- Case for Arithmetic Circuits}}
Recent results of Rotaru et al.~\cite{rotaru2017modes} showed that,
when considering \gls{MPC} over arithmetic circuits, efficient modes of
operation over encrypted data are possible if the underlying \gls{PRF} is MPC-friendly. We
integrate their approach~\cite{DBLP:conf/esorics/SymeonidisAMMDP17} into \system, and results from Table~\ref{tab:dabit}
show that it is at least $16$ times faster than using \gls{MPC} over binary
circuits with CBC-MAC-AES. This might come as a surprise because comparisons
are more expensive to do in arithmetic circuits. Recent improvements using edaBits~\cite{DBLP:conf/crypto/0001GKRS20} made comparisons much faster, which, in turn,
improved the \gls{MPC} protocols used by \system. To summarize, we breakdown the
cost into the following:
\begin{itemize}
    \item $10$ calls to MiMC to encrypt $M^{u_c}$ (excluding one call for
    computing the tweak according to \cite{rotaru2017modes}),
    \item $11$ more
    calls to compute $C^{u_c}$ - encrypting the concatenation
    of $AT^{veh_{u_o}}$ and $ID^{veh_{u_o}}$. Note that since we are using a
    different key than the first step we need to compute another tweak (one
    extra \gls{PRF} call),
    \item $6$ calls to compute $AuthTag^{BD^{u_o, u_c}}$, one more \gls{PRF} call for
    computing the tweak $N = E_{\vec{K}^{u_c}_{tag}[0]}(1)$ and a final \gls{PRF} call
    $E_{\vec{K}^{u_c}_{tag}[1]}(\mathsf{hash}'(ct))$ where $ct$ are the opened
    ciphertexts from encrypting $BD^{u_o, u_c}$ and $\mathsf{hash}'(\cdot))$ is a truncated
    version of a SHA-3 where we keep the first $128$ bits $\mathsf{hash}$~\cite{rotaru2017modes}.
\end{itemize}

If we include the \gls{PRF} calls to compute the tweaks, there are $31$ calls to a
\gls{PRF}, so one can think that the Boolean case is more efficient than the
arithmetic case. In practice, we see that HtMAC construction is faster than
CBC-MAC-AES, albeit with a factor of two communication overhead (see Table~\ref{tab:dabit}). One reason for this is that HtMAC is
fully parallelizable, resulting in an \gls{MPC} protocol with fewer rounds
than CBC-MAC-AES.

One of the main benefits of HtMAC construction is that it can be
instantiated with the Legendre \gls{PRF}, which can make the number of communication
even lower. We chose the MiMC based \gls{PRF} as that is demonstrated to be
faster on a LAN~\cite{rotaru2017modes} and to have a lower communication overhead -- although a higher number of communication rounds
than the Legendre-based \gls{PRF}~\cite{DBLP:conf/crypto/Damgard88}.

\subsection{Benchmark Results - Efficiency and Total Time}
Although our protocol's construction is agnostic to the underlying \gls{MPC}, its efficiency depends on the chosen \gls{MPC} scheme. We evaluate \system\ utilizing the semi-honest 3-party protocol by Araki et al.~\cite{DBLP:conf/ccs/ArakiFLNO16}. We evaluate its efficiency in both Boolean and arithmetic circuits with CBC-MAC-AES and HtMAC-MiMC, respectively. Note that we report timings for cryptographic operations and secure multiparty evaluations, leaving \gls{DB} accessing and communication latency on client-server and client-vehicle timings outside of our evaluations, as these are not dependent on the actual system construction.

\paragraph*{\textit{Step~1}}
Recalling Step~1 (see Fig.~\ref{fig:step1}), the operations of the \textit{session key generation and \gls{BD} sharing} is taking place on both users, the owner and consumer. At the consumer, $u_c$, the session key generation, using the $\mathsf{kdf}$ function, is implemented with  AES in CRT mode ($\approx \kdftime$~ms). The session keys are encrypted, using the $\mathsf{enc}$ function, with RSA-KEM specifications~\cite{rfc_5990} and $2048$-bit key-size ($\approx \Pubdectime$~ms). At the owner, $u_o$, the signature of the \gls{BD} is generated, using the $\mathsf{sign}$ function, with RSA-PKCS $\#1$ v2.0 specification~\cite{rfc_2437} and $2048$-bit output ($\approx \signtime$~ms). For the creation of the secret shares of $\mathsf{share}$ function, we implemented by the sharing primitive of Araki et al.~\cite{DBLP:conf/ccs/ArakiFLNO16} ($\approx \shareshamirtime$~ms). That results in a total estimation of $\approx \stepAtime$~ms in Step~1.

\paragraph*{\textit{Step~2}}
In Step~2, the \gls{AT} generation takes place at \gls{VSSP} (see Fig.~\ref{fig:step2}). We report the full range of experiments for a varying number of vehicles, $veh_{u_o}$ per owner, $u_o$ as it is illustrated in Fig.~\ref{fig:sims-er}: Fig.~\ref{fig:hermes-comm} regarding the intra-\gls{VSSP} communication cost, and Fig.~\ref{fig:hermes-tru} the throughput between servers.

Specifically, we vary the number of vehicle IDs (i.e., the numbers of vehicles registered per owner) and compute the communication rounds and data sent between
the \gls{VSSP} servers. We also compute the total throughput meaning the total number of \gls{AT}s generated per second (see Fig.~\ref{fig:step2}). In Table~\ref{tab:dabit},
we report the performance for a low number of vehicle IDs (i.e., $1,2,4$), representing individuals, but also for a large number of vehicles
(i.e., $256,512,1024$), representing (large branches of) vehicle-rental companies.

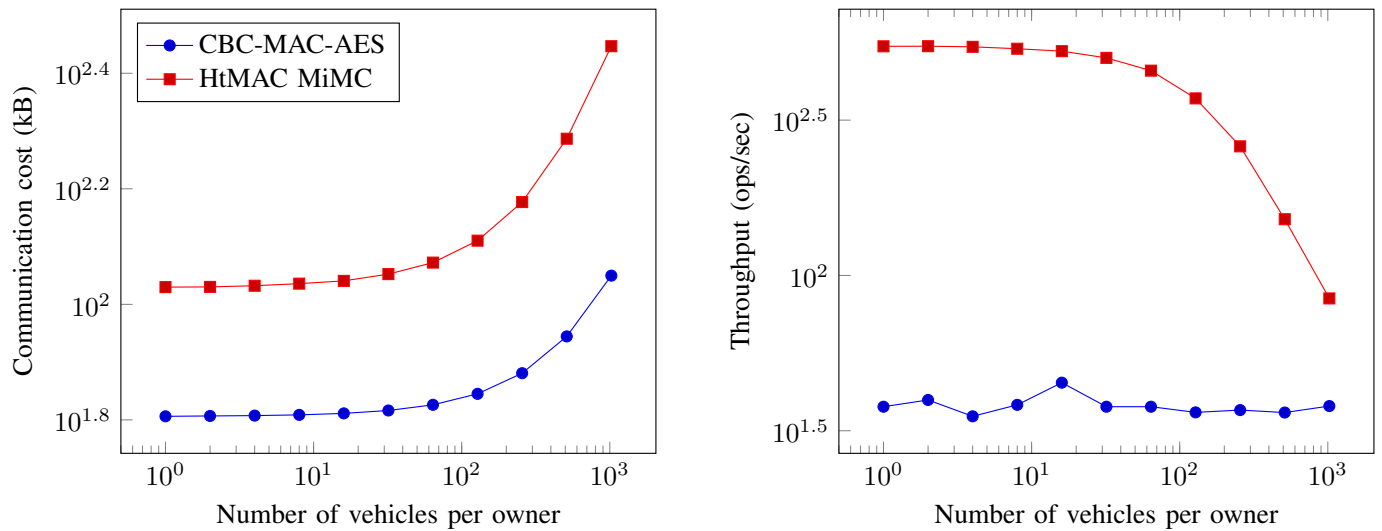
\begin{figure*}
    \centering
    \resizebox{\textwidth}{!}{%
        \begin{subfigure}[b]{0.5\textwidth}
    \begin{tikzpicture}
        \begin{axis}[
          tick pos = left,
          legend style   = {cells={anchor=west}},
          legend pos = north west,
          legend entries = {
            {CBC-MAC-AES},
            {HtMAC MiMC},
          },
          ylabel         = {Communication cost (kB)},
          xlabel         = {Number of vehicles per owner},
          ymode=log,
          xmode=log,
          log basis y={10},
          ]
          \addplot table[x=cars, y=comms] {experiments/cbc-comm};
          \addplot table[x=cars, y=comms] {experiments/htmac-comm};
        \end{axis}
    \end{tikzpicture}
    \caption{Communication cost per server: The intra-\gls{VSSP} communication cost for \acrfullpl{AT} generation (i.e., data sent-received).}
    \label{fig:hermes-comm}
    \end{subfigure}
        \begin{subfigure}[b]{0.45\textwidth}
   \begin{tikzpicture}
        \pgfplotsset{every axis legend/.append style={at={(10.5,10.1)},anchor=west}}
        \begin{axis}[
            ytick pos = left,
          ylabel         = {Throughput (ops/sec)},
          xlabel         = {Number of vehicles per owner},
          ymode=log,
          xmode=log,
          log basis y={10},
          ]
          \addplot table[x=cars, y=tru] {experiments/cbc-comm};
          \addplot table[x=cars, y=tru] {experiments/htmac-comm};
        \end{axis}
      \end{tikzpicture}
    \caption{Throughput for all servers: The throughput of all servers in \gls{VSSP} for \acrfullpl{AT} generation (i.e., ops) per second.}
    \label{fig:hermes-tru}
    \end{subfigure}
    }
    \caption{Communication cost and throughput at Step~2 (see Fig.~\ref{fig:step2}) for a variant number of vehicles per owner - form private individuals with a few vehicles, to rental companies with hundred or thousand of vehicles per branch.}
    \label{fig:sims-er}
\end{figure*}

We can see that the throughput of the \gls{AT} generation when
instantiated using CBC-MAC-AES remains constant, whereas, for HtMAC-MiMC, it is
decreasing. The reason for this is that when scaling up the number of vehicles, the
number of comparisons is increasing as well. For arithmetic
circuits, the comparisons become costly operations, whereas, for Boolean
circuits, comparisons can be made efficiently. However, the throughput for
HtMAC-MiMC is always better than CBC-MAC-AES, and this is because MiMC-based
\gls{PRF} is more lightweight -- requiring fewer multiplications -- and has a smaller
circuit depth.

\paragraph*{\textit{Step~3}}
The consumer in Step~3, queries, retrieves, verifies and decrypts the given \gls{AT} (see Fig.~\ref{fig:step3}). The verification of the \gls{AT} is implemented using the $\mathsf{mac}$ function ($\approx \MACvrftime$~ms). The total cost is $\approx \stepCtime$~ms in Step~3.

\paragraph*{\textit{Step~4}} 
The consumer delivers the \gls{AT} to \gls{OBU} of vehicle, which decrypts and verifies the signature in Step~4 (see Fig.~\ref{fig:step4}). Cryptographic operations are benchmarked at Nexcom \gls{OBU} box~\cite{nexcom:vtc6201-ft,preserve}. The decryption of \gls{AT} with the vehicle key, using the $\mathsf{D}$ function, is implemented with AES in CTR mode ($\approx \SymdecNexcomtime$~ms). The verification of signature of the \gls{BD}, using $\mathsf{verify}$ function, is implemented with RSA $2048$ ($\approx \signvrfNexcomtime$~ms). Finally, the signature is generated, using the $\mathsf{sign}$ function, with $2048$-bit output ($\approx \signNexcomtime$~ms). Note that the challenge-response protocol between the consumer and the vehicle does not directly affect the performance of \system, and thus we omit from our implementation and measurements. The total cost is $\approx \stepDtime$~ms in Step~4.

\paragraph*{\textit{Total}} 
The total cost of our cryptographic operations and \gls{MPC} evaluations considering the arithmetic circuits case (i.e., HtMAC-MiMC) is: {$\approx \totalMiMCsingleVehTime$}~ms for a single-vehicle owner, and {$\approx \totalMiMCmultiVehTime$}~ms for thousand vehicles per owner. It handles $546$ and $84$ access token generations per second, respectively. In addition, client-side \glspl{PD}, owner and consumer, and vehicle \glspl{OBU} need to perform only a few symmetric encryptions, signature and verification operations, making \system\ practical.

\subsection{Comparison with \exprotocol~\cite{DBLP:conf/esorics/SymeonidisAMMDP17}}

We report the main difference on efficiency and scalability between \system\ and \exprotocol~\cite{DBLP:conf/esorics/SymeonidisAMMDP17} is on Step~2 (see Table 2 in~\cite{DBLP:conf/esorics/SymeonidisAMMDP17}) -- the intra-\gls{VSSP} communication cost and throughput. \exprotocol\ reports $\approx 1.2$ seconds for generating the access token. When benchmarked on similar hardware we get a throughput of $33$ access tokens per second.~\footnote{\exprotocol\ specifications: Intel $i7$, $2.6$~Ghz CPU and $8$GB of RAM.} This makes \system\ with the CBC-MAC-AES construction roughly $\stepBCBCmoreAT$ times faster than \exprotocol. Switching from CBC-MAC-AES to HtMAC
offers a throughput of $546$ \glspl{AT} per second, which makes it $\approx
16.5$ times better than CBC-MAC, making it around $\stepBMiMCmoreAT$ times faster than
original timings in \exprotocol~\cite{DBLP:conf/esorics/SymeonidisAMMDP17}.
Thus, these results, specifically for Step~2 (see Fig.~\ref{fig:step2}), demonstrate the benefits of integrating our solution in a fully-fledged \gls{MPC} framework such as MP-SPDZ~\cite{DBLP:conf/ccs/Keller20}.
We stress that our implementation of \exprotocol\ was faster due to writing
CBC-MAC-AES using a mature \gls{MPC} framework such as MP-SPDZ rather than using
custom code as in \cite{DBLP:conf/esorics/SymeonidisAMMDP17}.

\subsection{Satisfying ESR1  --  Efficiency and scalability in a real-world deployment}
We demonstrate that \system\ maintains its efficiency, and it is scalable, supporting owners that could span from a few up to a thousand vehicles for (branches of) car-rental companies.

To argue about the real-world deployment aspect, we need first to find the answer to: ``\textit{how many vehicles per branch exist in a real-world deployment?}'' There is on average a few hundred (i.e., average $\approx~230$ / median $\approx~122$) of vehicles per branch in the U.S. in 2018~\cite{statistics:num-of-cars-per-carsharing} -- drawing from the analogy in \gls{VSS} of car-rental scenarios. It ranged from tens of vehicles (i.e., $\approx~29$) to an upper bound of almost a thousand (i.e., $\approx~900$) of vehicles per branch. Thus, is a safe approximation for \system\ supporting $1024$ vehicles per single owner (e.g., per branch), as in car-rental scenarios.

A follow up question is: ``\textit{how many daily vehicle-sharing operations are performed in \gls{VSS}?}'' This corresponds to the number of \gls{AT} generations in Step~2 (see Fig.~\ref{fig:step2}). According to reports~\cite{statistics:num-of-transactions-europe,statistics:num-of-transactions-avis-worldwide}, the total number of sharing operations of all car-rental transactions in Europe in 2017 is $86,41$~M ($\approx 237,000$ daily)~\cite{statistics:num-of-transactions-europe}. World-wide, the number of sharing operations compiles to $40$~M transactions in 2019 ($\approx 110,000$ daily) for Avis Budget group, one of the world-leading car-rental companies~\cite{statistics:num-of-transactions-avis-worldwide}.~\footnote{Assuming a uniform distribution for approximating the daily number of operations is reasonable.}

As \system\ supports a volume of $\approx~58,06$~M daily \glspl{AT} generation (see Table~\ref{tab:dabit}) -- considering the demanding scenario where an owner (i.e., a branch) shares a thousand vehicles -- our results show a two orders of magnitude more \glspl{AT} generation than the daily needs in real-world car-rental scenarios. Note that for comparison, we consider Step~2 computations for \glspl{AT} generation (see Fig.~\ref{fig:step2}) - intra-\gls{VSSP} computations can hinder the efficiency when scaling to multiple vehicles for a single owner.~\footnote{Recall that the costs are the non-linear operations such as comparisons over \gls{MPC}, the $\mathsf{eqz}$ function, to retrieve the vehicle keys in $\vec{D}^{u_o}$.} Thus, \system\ can scale and remain efficient, carrying millions of \glspl{AT} operations daily, capable of supporting a large number of vehicles per owner for short-term rental. Hence, it satisfies \textit{ESR1}.
 
\system\ straightforwardly can expand to support a vehicle-sharing company. Considering a single owner, as per branch-holder, a vehicle-sharing company can create multiple owners within the \system\ that each will manage their corresponding number of vehicles. Recall that each owner, branch-holder, can retrieve the corresponding set of their vehicles with a simple query at the DB, $DB^{S_i}$, that each \gls{VSSP} server holds. The query operation can be parallelizable, and its efficiency and scalability are related mainly by the underlying database structures and technologies, thus an orthogonal to \system.


\section{Related Work}\label{sec:related_work}

\begin{table*}[h]
\rowcolors{4}{}{gray!25}
  \centering
  \caption{Comparison of \system\ with state-of-the-art \acrfullpl{VSS} in terms of solution-design requirements, and \acrfull{SP} trust assumptions (see Sec.~\ref{sec:system_model}). Papers are listed in a chronological order.}
  \label{table:vehicle-access-comp}
  \resizebox{\textwidth}{!}{%
  \begin{tabular}{lccccccccccccccccc}
  \toprule  \toprule
  \multicolumn{1}{c}{\multirow{2}{*}{Paper}} &
    \multicolumn{2}{c}{Functional} &
    \multicolumn{8}{c}{Security} &
    \multicolumn{3}{c}{Privacy} &
    \multicolumn{2}{c}{Performance} &
    \multicolumn{2}{c}{\multirow{2}{*}{\gls{SP} assumption}} \\  \cmidrule{2-16} 
  \multicolumn{1}{c}{} &
    FR1 &
    FR2 &
    SR1 &
    SR2 &
    SR3 &
    SR4 &
    SR5 &
    SR6 &
    SR7 &
    SR8 &
    PR1 &
    PR2 &
    PR3 &
    1-1 &
    ESR1 & \\
    \toprule  \toprule
    Busold et al.~\cite{DBLP:conf/codaspy/BusoldTWDSSS13} & \checkmark & \checkmark & - & \checkmark & - & \checkmark & \checkmark & - & - & - & - & - & - & \checkmark & - & \textcolor{red}{Trusted} \\ [1ex]
    Kasper et.al.~\cite{DBLP:conf/rfidsec/KasperKOZP13} & \checkmark & - & \checkmark & \checkmark & \checkmark & - & - & \checkmark & \checkmark & \checkmark & - & - & - & \checkmark & - & \textcolor{red}{Trusted} \\ [1ex]
    Wei et al.~\cite{DBLP:journals/access/WeiYWWD17} & \checkmark & \checkmark & - & - & - & - & - & \checkmark & - & \checkmark & - & - & - & \checkmark & - & \textcolor{red}{Trusted} \\ [1ex]
    Groza et al.~\cite{DBLP:conf/vehits/GrozaAM17} & \checkmark & - & - & \checkmark & - & - & - & - & - & - & - & - & - & \checkmark & - & \textcolor{red}{Trusted} \\ [1ex]
    Dmitrienko and Plappert~\cite{DBLP:conf/codaspy/DmitrienkoP17} & \checkmark & - & \checkmark & \checkmark & \checkmark & - &  & \checkmark & - & \checkmark & - & - & - & \checkmark & - & \textcolor{red}{Trusted} \\ [1ex]
    \exprotocol~\cite{DBLP:conf/esorics/SymeonidisAMMDP17} & \checkmark & \checkmark & \checkmark & \checkmark & \checkmark & \checkmark & \checkmark & \checkmark & \checkmark & \checkmark & \checkmark & \checkmark & \checkmark & \checkmark & - & UnTrusted \\ [1ex]
    Groza et al.~\cite{DBLP:journals/access/GrozaABMG20} & \checkmark & \checkmark & - & \checkmark & \checkmark & \checkmark & \checkmark & \checkmark & - & \checkmark & - & - & - & \checkmark & - & \textcolor{red}{Trusted} \\ [1ex]
    \textbf{\system}  & \checkmark & \checkmark & \checkmark & \checkmark & \checkmark & \checkmark & \checkmark & \checkmark & \checkmark & \checkmark & \checkmark & \checkmark & \checkmark & \checkmark & \checkmark & UnTrusted \\ [1ex]
  \bottomrule \bottomrule
  \end{tabular}%
  }
  \end{table*}


State of the art on \acrfullpl{VSS} ranges from (fully) trusting \acrfullpl{SP} to consider them having an adversarial behavior. Design assumptions on trust affect the selected requirements and, subsequently, solution designs. As illustrated in Table~\ref{table:vehicle-access-comp}, there is a large body of work for secure vehicle access and sharing in \glspl{VSS}. However, users' privacy towards an untrusted \gls{SP} is only considered by~\cite{DBLP:conf/esorics/SymeonidisAMMDP17} and \system, with the current system design advancing significantly~\cite{DBLP:conf/esorics/SymeonidisAMMDP17} in terms of efficiency and scalability. 

All other proposed solution designs for vehicle accessing and delegation are considering \gls{SP} as a trusted entity to collect data for access and sharing operations in \gls{VSS}. There can have control over users' data by generating and storing session keys of transactions and master keys of vehicles. Initially, Busold et al.~\cite{DBLP:conf/codaspy/BusoldTWDSSS13} proposed a protocol for dynamic access to a car's immobilizer and delegation possibilities for accessing. At their proposed protocol, the vehicle owner and \gls{VM} exchange keys used to encrypt and sign two \acrfullpl{AT} -- one for authenticating the owner accessing the car and one for delegating access rights to a consumer. Confidentiality of \gls{BD} and \gls{AT} is not preserved, as the delegation is happening using a \gls{MAC}-signing operation. Accountability, non-reputation of origin, and delivery of \gls{AT} is also not preserved due to \gls{MAC}-signing - the session key is generated by the owner for each delegation operation. \cite{DBLP:conf/codaspy/BusoldTWDSSS13} treats the \gls{VM} as a trusted holding session keys for authentication and access of the vehicle. 

The work of Kasper et. al.~\cite{DBLP:conf/rfidsec/KasperKOZP13} considers a trusted car-sharing \gls{SP} and a eID. The eID interacts with a user to register the user at the vehicle-sharing \gls{SP}. Their solution design considers public keys for encryption and signing a delegation compromising backward and forward secrecy along with privacy requirements. Wei et al.~\cite{DBLP:journals/access/WeiYWWD17} offers a similar solution to~\cite{DBLP:conf/rfidsec/KasperKOZP13} using identity-based encryption for the generation of the public/private key pairs for the owner and consumer. The keys are generated with the identity of the owner and its car. For the consumer, the inputs are the customer's identity and the access rights granted for the vehicle. The \gls{BD} is sent in the clear, lacking data and entity authentication verification by the vehicle.

Groza et al.~\cite{DBLP:conf/vehits/GrozaAM17} proposed an access control and delegation of rights protocol using an MSP430 Microcontroller. Their main security operations for delegation are to provide data authentication of \gls{BD} using \gls{MAC} cryptographic primitive. Dmitrienko and Plappert~\cite{DBLP:conf/codaspy/DmitrienkoP17} designated a secure free-floating vehicle sharing system. They proposed using two-factor authentication, RFID-enabled smart cards, and \Glspl{PD} to access a vehicle. In contrast to \system, their solution design considers a centralized vehicle-sharing \gls{SP} that is fully trusted. The \gls{SP} has access to the master key of vehicles in clear contrast to \system. Thus, it allows the \gls{SP} to collect the information exchanged between the vehicle, the \gls{SP} and each of the users for every vehicle access provision.

In recent work, Groza et al.~\cite{DBLP:journals/access/GrozaABMG20} proposed an access control protocol using smartphones as a vehicle key for vehicle access and delegation enabling sharing. To preserve the security and anonymity of users accessing a vehicle, they combined identity-based encryption and group signatures. In their solution design~\cite{DBLP:journals/access/GrozaABMG20}, they distinguished two types of sharing, persistent and ephemeral. Although they utilize group signatures for privacy preservation, it is only for persistent delegation. In ephemeral delegation for dynamic vehicle-sharing, identity encryption is used, removing the anonymity properties that group signatures can apply. Hence, we consider their solution as only a secure approach to vehicle sharing.

\exprotocol~\cite{DBLP:conf/esorics/SymeonidisAMMDP17} improves on the work proposed in~\cite{DBLP:conf/codaspy/DmitrienkoP17} in terms of the adversarial consideration of the \gls{SP} (i.e., \gls{VSSP}), the privacy requirements, and the secrecy of vehicle keys towards the \gls{SP}, to mention a few. In specific, it considers untrusted servers in \gls{VSSP} for the generation and distribution of \glspl{AT}. The authors utilize \gls{MPC} in combination with several cryptographic primitives. With their work, they also consider malicious users and support user accountability, revealing a user's identity in wrongdoings. However, \exprotocol\ is not tested on how it scales to multiple evaluations -- to a large fleet of the vehicle with multiple owners of multiple vehicles. \system\  maintains the design advantages of \exprotocol~\cite{DBLP:conf/esorics/SymeonidisAMMDP17}, and is proven to run significantly faster than~\cite{DBLP:conf/esorics/SymeonidisAMMDP17} due to its optimized design and \gls{MPC} constructions.

Work on vehicle-sharing also focuses on complementary operations to access provisions, such as booking, payments, and accountability. Huang et al.~\cite{DBLP:journals/tvt/HuangLNS20} proposed a privacy-preserving identity management protocol focusing on authentication while verifying users' misbehavior. They utilize decentralized entities and a centralized vehicle sharing \gls{SP}. However, the \gls{SP} is trusted and can know who is sharing, which vehicle, with whom. Madhusudan~et~al.~\cite{DBLP:conf/icissp/MadhusudanSMZP19} and De Troch~\cite{de2020dpace} proposed privacy-preserving protocols for booking and payment operations on vehicle sharing systems. Their protocols utilize smart contracts on the Ethereum blockchain. Trust is placed on cryptographic primitives and blockchain instead of a centralized \gls{SP}. De Troch~\cite{de2020dpace} also considers accountability in case of misbehavior, in which there is a loss of privacy and deposit to punish malicious behavior. 

Beyond vehicle sharing security and privacy, vehicular communications security and privacy received extensive attention over the years~\cite{DBLP:journals/tits/KhodaeiJP18,9311462,PapadimitratosGH:C:2006}. Recent results focus, for example, on scalable systems, notably for credential management~\cite{DBLP:journals/tits/KhodaeiJP18,DBLP:conf/wisec/KhodaeiNP19,DBLP:journals/corr/abs-2004-03407}, and decentralized cooperative defenses~\cite{DBLP:journals/tissec/JinP19,DBLP:journals/adhoc/JinP19}. Moreover, Huayi et al.~\cite{qi2020scalable} proposed an enhanced scheme of~\cite{DBLP:journals/tdsc/TroncosoDKBP11}, namely DUBI, a decentralized and privacy-preserving usage-based insurance scheme built on the blockchain technology to address privacy concerns for pay-as-you-drive insurances using zero-knowledge proofs and smart contracts.

\section{Conclusion}\label{sec:conclusion}
In this paper, we proposed \system\ -- an efficient, scalable, secure, and privacy-enhancing system for vehicle access provision. It allows users to dynamically instantiate, share, and access vehicles in a secure and privacy-enhancing fashion. To achieve its security and privacy guarantees, \system\ deploys secure multiparty computation for access token generation and sharing while keeping transactions and booking details confidential. To ensure efficiency and scalability, \system\ utilizes cryptographic primitives in combination with secure multiparty-computation protocols, supporting various users and vehicles per user. We presented a formal analysis of our system security and privacy requirements and designed a prototype as a proof-of-concept. 

We demonstrated that \system\ is suitable for serving large numbers of individuals, each with few vehicles and rental companies with hundred or thousand of vehicles per branch. We benchmarked the cryptographic operations and secure multiparty evaluations testing over arithmetic circuits with HtMAC-MiMC demonstrating its efficiency and scalability. For comparison to \exprotocol, we tested \system\ for the case of binary circuits with CBC-MAC-AES. We showed that \system\ achieves a significant performance improvement: $\approx \stepBmacCBCtime$~ms for a vehicle access provision, thus demonstrating its efficiency compared to~\cite{DBLP:conf/esorics/SymeonidisAMMDP17} (i.e., $\stepBCBCmoreAT$ times faster). We also demonstrated that \system\ is practical on the vehicle side too as a \gls{AT} operations on a prototype \gls{OBU} box takes only $\approx \stepDtime$~ms. 

In the future, aiming to make the operations even more efficient, we will investigate cryptographic primitives using lightweight block ciphers such as Rasta. We also plan to extend \system\ to booking and payment operations and protect against active adversaries on the untrusted servers.

\bibliographystyle{IEEEtran}
\bibliography{biblio.bib}


\ifCLASSOPTIONcaptionsoff
  \newpage
\fi




%
\appendices

\section{\system\ complete representation and simplified representation for the proof of Theorem~\ref{thmextended}.}\label{appendix:protocol_complete}

We provide complete representation of  \system\ including all cryptographic operations and messages exchanged for Step~1~-~Step~4. Moreover, we provide a simplified representation of \system\ for the proof of Theorem~\ref{thmextended}.


\begin{landscape}


    \begin{figure}[!ht]
\centering
\resizebox{1.25\textwidth}{!}{%
\fbox{
\begin{tikzpicture}
    \node (o) at (0,0) {\textbf{Owner ($u_o$)}};
    \node[below = \y em of o] (odot) {};
    \draw[dotted] (o)--(odot);
    
    \node[right = 8em of o] (vehicle) {\textbf{Vehicle ($veh_{u_o}$)}};
    \node[below = \y em of vehicle] (vehicledot) {};
    \draw[dotted] (vehicle)--(vehicledot);
    
    \node[right = 9.5em of vehicle] (consumer) {\textbf{Consumer ($u_c$)}};
    \node[below = \y em of consumer] (cdot) {};
    \draw[dotted] (consumer)--(cdot);
    
    \node[right = 5.5em of consumer] (pl) {\textbf{Public Ledger ($\mathsf{PL}$)}};
    \node[below = \y em of pl] (pdot) {};
    \draw[dotted] (pl)--(pdot);
    
    \node[right = 13em of pl] (s) {\textbf{Servers} $\mathsf{S}_1\dots \mathsf{S}_i\dots \mathsf{S}_l$};
    \node[below = \y em of s] (sdot) {};
    \draw[dotted] (s)--(sdot);
    
    \node[below = 1em of o] (o11) {};
    \node[right = 31em of o11] (consumer11) {};
    \draw [<->] (o11) -- node[fill=white] {$BD^{u_o, u_c} = \{\mathsf{hash}(\mathit{Cert}^{u_c}) ,ID^{veh_{u_o}},L^{veh_{u_o}},CD^{u_c},AC^{u_c},ID^{BD}\}$} (consumer11);
    
    \node[below = 3em of o] (o1) {};
    \node[right = 31em of o1] (consumer1) {};
    \draw [->] (o1) -- node[fill=white] {msg$\{SES\_K\_GEN\_REQ, ID^{BD}\}$} (consumer1);
    
    \node[below = 1em of consumer1, fill=white, draw, rounded corners] (prf1) { 
                \begin{varwidth}{\linewidth}
                    \begin{algorithmic}[1]
                        \STATE $\{K^{u_c}_{enc}, \vec{K}^{u_c}_{tag}\} \leftarrow \mathsf{kdf}(K^{u_c}_{master}, counter)$
                            \STATE $[K^{u_c}_{enc}] \leftarrow \mathsf{share}(K^{u_c}_{enc})$
                            \STATE $[\vec{K}^{u_c}_{tag}] \leftarrow \mathsf{share}(\vec{K}^{u_c}_{tag})$
                        \FOR{$i = 1\dots l$}
                            \STATE $C^{\mathsf{S}_i} \leftarrow \mathsf{enc}(Pk^{\mathsf{S}_i}, \{[K^{u_c}_{enc}],[\vec{K}^{u_c}_{tag}]\})$
                        \ENDFOR
            \end{algorithmic}%
            \end{varwidth}
            };

    \node[below = 1em of o1, fill=white, draw, rounded corners] (mpc3) {
            \begin{varwidth}{\linewidth}
                \begin{algorithmic}[1]
                    \STATE $\sigma^{u_o} \leftarrow \mathsf{sign}(Sk^{u_o},BD^{u_o, u_c})$
                    \STATE $M^{u_c} \leftarrow \{BD^{u_o, u_c}, \sigma^{u_o}\}$
                    \STATE $[M^{u_c}] \leftarrow \mathsf{share}(M^{u_c})$
                \end{algorithmic}%
            \end{varwidth}
            };  
            
    \node[below = 10em of consumer1] (consumer2) {};
    \node[left = 31em of consumer2] (o2) {};
    \draw [->] (consumer2) -- node[fill=white] {msg$\{SES\_K\_GEN\_ACK, ID^{BD}, \{C^{S_1}, \dots, C^{S_l}\}\}$} (o2);
    
    \node[below = 1em of o2] (o3) {};
    \node[right = 67em of o3] (s1) {};
    \draw [->] (o3) -- node[fill=white] {msg$_i\{AT\_GEN\_REQ, ID^{u_o}, C^{\mathsf{S}_i}, [M^{u_c}]\}$} (s1);

    \node[below = 1em of s1, fill=white, draw, rounded corners] (s2) { 
                \begin{varwidth}{\linewidth}
                    \begin{algorithmic}[1]
                        \STATE $\{[K^{u_c}_{enc}],[\vec{K}^{u_c}_{tag}]\} \leftarrow \mathsf{dec}(Sk^{\mathsf{S}_i},C^{\mathsf{S}_i})$
                        \STATE $\vec{D}^{u_o} \leftarrow \mathsf{query}(ID^{u_o}, DB^{\mathsf{S}_i})$
                        \FOR{$y = 1\dots n$}
                            \STATE $\vec{[D]}^{u_o} \leftarrow ([ID^{veh_{u_o}}] \stackrel{?}{=} [ID^{veh_{u_o}}_y])$
                        \ENDFOR
                        \STATE $[K^{veh_{u_o}}] \leftarrow \vec{D}^{veh_{u_o}} \times \vec{D}^{u_o}$
                        \STATE $[AT^{veh_{u_o}}] \leftarrow \mathsf{E}([K^{veh_{u_o}}], [M^{u_c}])$
                        \STATE $[C^{u_c}] \leftarrow \mathsf{E}([K^{u_c}_{enc}],\{[AT^{veh_{u_o}}], [ID^{veh_{u_o}}]\})$
                        \STATE $C^{u_c} \leftarrow \mathsf{open}([C^{u_c}])$
                        \STATE $[AuthTag^{BD^{u_o, u_c}}] \leftarrow \mathsf{mac}([K^{u_c}_{tag_{mac}}], E(K^{u_c}_{tag_{enc}}, [BD^{u_o, u_c}]))$
                        \STATE $AuthTag^{BD^{u_o, u_c}} \leftarrow \mathsf{open}([AuthTag^{BD^{u_o, u_c}}])$
                \end{algorithmic}%
            \end{varwidth}
            };
            
    \node[below = 1em of s2] (s3) {};
    \node[left = 22em of s3] (pl1) {};
    \draw [->] (s3) -- node[fill=white] {msg$_i\{AT\_PUB\_REQ, C^{u_c}, AuthTag^{BD^{u_o, u_c}}\}$} (pl1);

    \node[below = 1em of pl1, fill=white, draw, rounded corners] (pl2) {
        \begin{varwidth}{\linewidth}
        \begin{algorithmic}
            \STATE $\mathsf{publish}(TS^{Pub}_i, C^{u_c}, AuthTag^{BD^{u_o, u_c}})$
        \end{algorithmic}%
        \end{varwidth}
        };
      
    \node[below = 1em of pl2] (pl3) {};
    \node[right = 22em of pl3] (s4) {};
    \draw [->] (pl3) -- node[fill=white] {msg$\{M\_PUB\_ACK, TS^{Pub}_i\}$} (s4);

    \node[below = 1em of s4] (s5) {};
    \node[left = 67em of s5] (o4) {};
    \draw [->] (s5) -- node[fill=white] {msg$\{AT\_PUB\_ACK, TS^{Pub}_i\}$} (o4);
    
    \node[below = 1em of o4] (o5) {};
    \node[right =31em of o5] (consumer3) {};
    \draw [->] (o5) -- node[fill=white] {msg$\{AT\_PUB\_ACK, TS^{Pub}_i\}$} (consumer3);
    
    \node[below = 5em of pl3] (pl4) {};
    \draw (pl4) node[fill=white] {
    \begin{tabularx}{9.5cm}{|X|X|X|}
        \hline
        $TS^{Pub}_i$ & $C^{u_c}$ & $AuthTag^{BD^{u_o, u_c}}$ \\
        \hline
        14774098 & ersdf3tx0 & fwefw234 \\
        \hline
        $\dots$ & $\dots$ & $\dots$ \\
        \hline
      \end{tabularx}
      };
    
    \node[below = 5em of consumer3] (consumer4) {};
    \node[right = 13em of consumer4] (pl5) {};
    \draw [->] (consumer4) -- node[fill=white] {$\mathsf{query\_an}(TS^{Pub}_i)$} (pl5);

    \node[below = 1em of pl5] (pl6) {};
    \node[left = 13em of pl6] (consumer5) {};
    \draw [->] (pl6) -- node[fill=white] {msg$\{C^{u_c}, AuthTag^{BD^{u_o, u_c}}\}$} (consumer5);
    
    \node[below = 1em of consumer5, fill=white, draw, rounded corners] (consumer6) {
        \begin{varwidth}{\linewidth}
        \begin{algorithmic}
        \IF {$AuthTag^{BD^{u_o, u_c}} \stackrel{?}{=} \mathsf{mac}(K^{u_c}_{tag_{mac}}, E(K^{u_c}_{tag_{enc}}, BD^{u_o, u_c}))$}
            \STATE $\{AT^{veh_{u_o}}, ID^{veh_{u_o}}\} \leftarrow \mathsf{D}(K^{u_c}_{enc},C^{u_c})$
        \ELSE
            \STATE Break
        \ENDIF
        \end{algorithmic}%
        \end{varwidth}
        };   
      
    \node[below = 1em of consumer6] (consumer7) {};
    \node[left = 16.5em of consumer7] (c1) {};
    \draw [dashed,->] (consumer7) -- node[fill=white] {msg$\{AT^{veh_{u_o}}, ID^{veh_{u_o}},\mathit{Cert}^{u_c}\}$} (c1);
    
    \node[below = 1em of c1, fill=white, draw, rounded corners] (c2) {
        \begin{varwidth}{\linewidth}
            \begin{algorithmic}[1]
                \STATE $\{BD^{u_o, u_c}, \sigma^{u_o}\} \leftarrow \mathsf{D}(K^{veh_{u_o}},AT^{veh_{u_o}})$
                \STATE $\mathsf{verify}(Pk^{u_o},BD^{u_o, u_c},\sigma^{u_o})$
            \end{algorithmic}%
        \end{varwidth}
        }; 
    
    \node[below = 5em of c1] (c2) {};
    \node[right = 16em of c2] (consumer8) {};
    \draw [dashed,<->] (c2) -- node[fill=white] {Challenge / Response} (consumer8);
    
    \node[below = 1em of c2, fill=white, draw, rounded corners] (c2) {
    \begin{varwidth}{\linewidth}
         \begin{algorithmic}
             \STATE $\sigma^{veh_{u_o}}_{Access} \leftarrow \mathsf{sign}(Sk^{veh_{u_o}}, \{BD^{u_o, u_c},TS^{veh_{u_o}}_{Access}\})$
         \end{algorithmic}%
         \end{varwidth}
         }; 
        
    \node[below = 1em of c2] (c3) {};
    \node[left = 14em of c3] (o6) {};
    \draw [dashed, ->] (c3) -- node[fill=white] {msg$\{\sigma^{veh_{u_o}}_{Access},TS^{veh_{u_o}}_{Access}\}$} (o6);
    
    \node[below = 1em of o6, fill=white, draw, rounded corners] (o7) {
        \begin{varwidth}{\linewidth}
        \begin{algorithmic}
            \STATE $\mathsf{verify}(Pk^{veh_{u_o}},\{BD^{u_o, u_c},TS^{veh_{u_o}}_{Access}\}, \sigma^{veh_{u_o}}_{Access})$
        \end{algorithmic}%
        \end{varwidth}
        };
    \end{tikzpicture}}}
    \caption{\system\ complete representation.}
    \label{fig:prot}
    \end{figure}
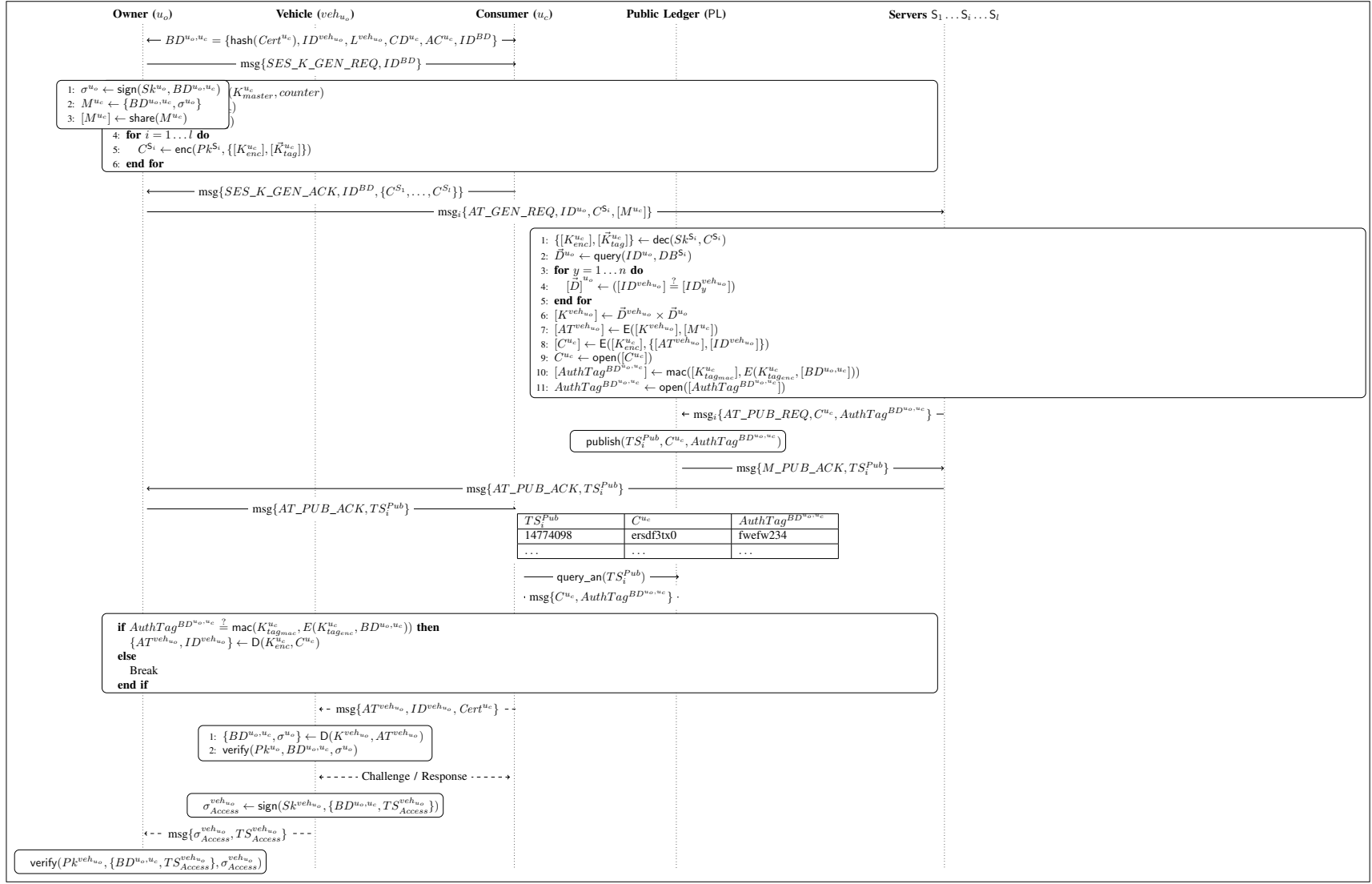
\end{landscape}


\newcommand\yproof{51.5}

\begin{landscape}
\begin{figure}[!ht]
\centering
\resizebox{1.25\textwidth}{!}{%
\fbox{
\begin{tikzpicture}
    \node (o) at (0,0) {\textbf{Owner ($u_o$)}};
    \node[below = \sy em of o] (odot) {};
    \draw[dotted] (o)--(odot);
    
    \node[right = 8em of o] (vehicle) {\textbf{Vehicle ($veh_{u_o}$)}};
    \node[below = \sy em of vehicle] (vehicledot) {};
    \draw[dotted] (vehicle)--(vehicledot);
    
    \node[right = 9.5em of vehicle] (consumer) {\textbf{Consumer ($u_c$)}};
    \node[below = \sy em of consumer] (cdot) {};
    \draw[dotted] (consumer)--(cdot);
    
    \node[right = 5.5em of consumer] (pl) {\textbf{Public Ledger ($\mathsf{PL}$)}};
    \node[below = \sy em of pl] (pdot) {};
    \draw[dotted] (pl)--(pdot);
    
    \node[right = 13.5em of pl] (s) {\textbf{\gls{VSSP} (trusted)}};
    \node[below = \sy em of s] (sdot) {};
    \draw[dotted] (s)--(sdot);
    
    \node[below = 1em of o] (o11) {};
    \node[right = 31em of o11] (consumer11) {};
    \draw [<->] (o11) -- node[fill=white] {$BD^{u_o, u_c} = \{\mathsf{hash}(\mathit{Cert}^{u_c}) ,ID^{veh_{u_o}},L^{veh_{u_o}},CD^{u_c},AC^{u_c},ID^{BD}\}$} (consumer11);
    
    \node[below = 3em of o] (o1) {};
    \node[right = 31em of o1] (consumer1) {};
    \draw [->] (o1) -- node[fill=white] {msg$\{SES\_K\_GEN\_REQ, ID^{BD}\}$} (consumer1);
    
    \node[below = 1em of consumer1, fill=white, draw, rounded corners] (prf1) { 
                \begin{varwidth}{\linewidth}
                    \begin{algorithmic}[1]
                        \STATE $\{\pi^{u_c}_{enc},\pi^{u_c}_{tag_{enc}}\} \getsR \mathrm{Func}(\mathsf{E})$
                        \STATE $K^{u_c}_{tag_{mac}} \leftarrow \$^{u_c}$
                        \STATE $C^{VSSP} \leftarrow \rho^{VSSP}(\{\pi^{u_c}_{enc},\pi^{u_c}_{tag_{enc}},K^{u_c}_{tag_{mac}}\})$
            \end{algorithmic}%
            \end{varwidth}
            };

    \node[below = 1em of o1, fill=white, draw, rounded corners] (mpc3) {
            \begin{varwidth}{\linewidth}
                \begin{algorithmic}[1]
                    \STATE $\sigma^{u_o} \leftarrow \mathsf{sign}(Sk^{u_o},BD^{u_o, u_c})$
                    \STATE $M^{u_c} \leftarrow \{BD^{u_o, u_c}, \sigma^{u_o}\}$
                \end{algorithmic}%
            \end{varwidth}
            };  
            
    \node[below = 10em of consumer1] (consumer2) {};
    \node[left = 31em of consumer2] (o2) {};
    \draw [->] (consumer2) -- node[fill=white] {msg$\{SES\_K\_GEN\_ACK, ID^{BD}, C^{VSSP}\}$} (o2);
    
    \node[below = 1em of o2] (o3) {};
    \node[right = 67em of o3] (s1) {};
    \draw [->] (o3) -- node[fill=white] {msg$_i\{AT\_GEN\_REQ, ID^{u_o}, C^{VSSP}, M^{u_c}\}$} (s1);

    \node[below = 1em of s1, fill=white, draw, rounded corners] (s2) { 
                \begin{varwidth}{\linewidth}
                    \begin{algorithmic}[1]
                        \STATE $\{\pi^{u_c}_{enc},\pi^{u_c}_{tag_{enc}},K^{u_c}_{tag_{mac}}\} \leftarrow (\rho^{VSSP})^{-1}(C^{VSSP})$
                        \STATE $\pi^{veh_{u_o}} \leftarrow \mathsf{query}(ID^{u_o}, DB^{VSSP})$
                        \STATE $AT^{veh_{u_o}} \leftarrow \pi^{veh_{u_o}}(M^{u_c})$
                        \STATE $C^{u_c} \leftarrow \pi^{u_c}_{enc}(\{AT^{veh_{u_o}}, ID^{veh_{u_o}}\})$
                        \STATE $AuthTag^{BD^{u_o, u_c}} \leftarrow \mathsf{mac}(K^{u_c}_{tag_{mac}}, \pi^{u_c}_{tag_{enc}}(BD^{u_o, u_c}))$
                \end{algorithmic}%
            \end{varwidth}
            };
            
    \node[below = 1em of s2] (s3) {};
    \node[left = 22em of s3] (pl1) {};
    \draw [->] (s3) -- node[fill=white] {msg$_i\{AT\_PUB\_REQ, C^{u_c}, AuthTag^{BD^{u_o, u_c}}\}$} (pl1);

    \node[below = 1em of pl1, fill=white, draw, rounded corners] (pl2) {
        \begin{varwidth}{\linewidth}
        \begin{algorithmic}
            \STATE $\mathsf{publish}(TS^{Pub}_i, C^{u_c}, AuthTag^{BD^{u_o, u_c}})$
        \end{algorithmic}%
        \end{varwidth}
        };
      
    \node[below = 1em of pl2] (pl3) {};
    \node[right = 22em of pl3] (s4) {};
    \draw [->] (pl3) -- node[fill=white] {msg$\{M\_PUB\_ACK, TS^{Pub}_i\}$} (s4);

    \node[below = 1em of s4] (s5) {};
    \node[left = 67em of s5] (o4) {};
    \draw [->] (s5) -- node[fill=white] {msg$\{AT\_PUB\_ACK, TS^{Pub}_i\}$} (o4);
    
    \node[below = 1em of o4] (o5) {};
    \node[right =31em of o5] (consumer3) {};
    \draw [->] (o5) -- node[fill=white] {msg$\{AT\_PUB\_ACK, TS^{Pub}_i\}$} (consumer3);
    
    \node[below = 5em of pl3] (pl4) {};
    \draw (pl4) node[fill=white] {
    \begin{tabularx}{9.5cm}{|X|X|X|}
        \hline
        $TS^{Pub}_i$ & $C^{u_c}$ & $AuthTag^{BD^{u_o, u_c}}$ \\
        \hline
        14774098 & ersdf3tx0 & fwefw234 \\
        \hline
        $\dots$ & $\dots$ & $\dots$ \\
        \hline
      \end{tabularx}
      };
    
    \node[below = 5em of consumer3] (consumer4) {};
    \node[right = 13em of consumer4] (pl5) {};
    \draw [->] (consumer4) -- node[fill=white] {$\mathsf{query\_an}(TS^{Pub}_i)$} (pl5);

    \node[below = 1em of pl5] (pl6) {};
    \node[left = 13em of pl6] (consumer5) {};
    \draw [->] (pl6) -- node[fill=white] {msg$\{C^{u_c}, AuthTag^{BD^{u_o, u_c}}\}$} (consumer5);
    
    \node[below = 1em of consumer5, fill=white, draw, rounded corners] (consumer6) {
        \begin{varwidth}{\linewidth}
        \begin{algorithmic}
        \IF {$AuthTag^{BD^{u_o, u_c}} \stackrel{?}{=} \mathsf{mac}(K^{u_c}_{tag_{mac}}, E(K^{u_c}_{tag_{enc}}, BD^{u_o, u_c}))$}
            \STATE $\{AT^{veh_{u_o}}, ID^{veh_{u_o}}\} \leftarrow (\pi^{u_c}_{enc})^{-1}(C^{u_c})$
        \ELSE
            \STATE Break
        \ENDIF
        \end{algorithmic}%
        \end{varwidth}
        };   
      
    \node[below = 1em of consumer6] (consumer7) {};
    \node[left = 16.5em of consumer7] (c1) {};
    \draw [dashed,->] (consumer7) -- node[fill=white] {msg$\{AT^{veh_{u_o}}, ID^{veh_{u_o}},\mathit{Cert}^{u_c}\}$} (c1);
    
    \node[below = 1em of c1, fill=white, draw, rounded corners] (c2) {
        \begin{varwidth}{\linewidth}
            \begin{algorithmic}[1]
                \STATE $\{BD^{u_o, u_c}, \sigma^{u_o}\} \leftarrow (\pi^{veh_{u_o}})^{-1}(AT^{veh_{u_o}})$
                \STATE $\mathsf{verify}(Pk^{u_o},BD^{u_o, u_c},\sigma^{u_o})$
            \end{algorithmic}%
        \end{varwidth}
        }; 
    
    \node[below = 5em of c1] (c2) {};
    \node[right = 16em of c2] (consumer8) {};
    \draw [dashed,<->] (c2) -- node[fill=white] {Challenge / Response} (consumer8);
    
    \node[below = 1em of c2, fill=white, draw, rounded corners] (c2) {
    \begin{varwidth}{\linewidth}
         \begin{algorithmic}
             \STATE $\sigma^{veh_{u_o}}_{Access} \leftarrow \mathsf{sign}(Sk^{veh_{u_o}}, \{BD^{u_o, u_c},TS^{veh_{u_o}}_{Access}\})$
         \end{algorithmic}%
         \end{varwidth}
         }; 
        
    \node[below = 1em of c2] (c3) {};
    \node[left = 14em of c3] (o6) {};
    \draw [dashed, ->] (c3) -- node[fill=white] {msg$\{\sigma^{veh_{u_o}}_{Access},TS^{veh_{u_o}}_{Access}\}$} (o6);
    
    \node[below = 1em of o6, fill=white, draw, rounded corners] (o7) {
        \begin{varwidth}{\linewidth}
        \begin{algorithmic}
            \STATE $\mathsf{verify}(Pk^{veh_{u_o}},\{BD^{u_o, u_c},TS^{veh_{u_o}}_{Access}\}, \sigma^{veh_{u_o}}_{Access})$
        \end{algorithmic}%
        \end{varwidth}
        };
    \end{tikzpicture}}}
    \caption{Simplified representation of \system\ for the proof of Theorem~\ref{thmextended}.}
    \label{fig:protsimplified}
    \end{figure}
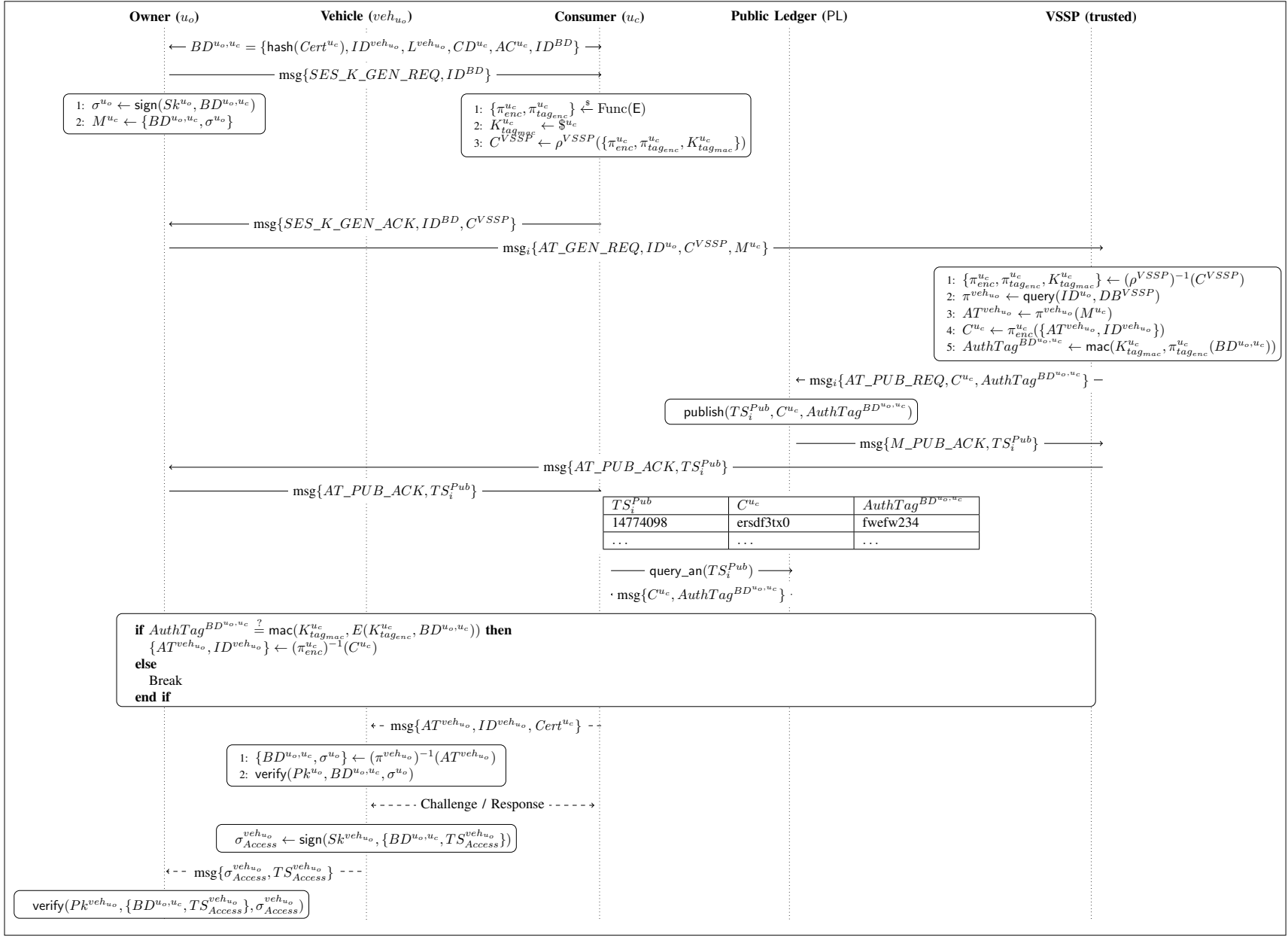
\end{landscape}


\end{document}